\documentclass[11pt]{article}

\usepackage[english]{babel}

\usepackage[a4paper, margin=1in]{geometry}
\usepackage{amsmath}
\usepackage{amsfonts}
\usepackage{graphicx}
\usepackage{color}
\usepackage[colorlinks=true, allcolors=blue]{hyperref}
\usepackage{standalone}%\includestandalone[mode=buildnew, width = 5cm]{figure}

\usepackage[textsize=tiny, color=cyan]{todonotes}

\newcommand{\Hquad}{\hspace{0.5em}} % Half space quad

\usepackage{amsthm}
\usepackage{thmtools, thm-restate} % re-statable theorems
\newtheorem{theorem}{Theorem}
\newtheorem{corollary}[theorem]{Corollary}
\newtheorem{lemma}[theorem]{Lemma}

\usepackage{algorithm}
\usepackage[]{algpseudocode}

\usepackage{booktabs}

\usepackage[
backend=biber,
style=alphabetic,
maxbibnames=8
]{biblatex}
\AtEveryBibitem{
    \clearfield{urlyear}
    \clearfield{urlmonth}
    \clearfield{url}
    \clearfield{note}
    \clearfield{month}
    \clearfield{day}
    \clearfield{editor}
    \clearfield{editorb}
    \clearfield{location}
    \clearfield{extra}
    \clearfield{isbn}
    \clearfield{issn}
    \clearlist{language}
    \clearname{editor}
    \clearname{editorb}
    \clearname{editorc}
}

\usepackage{framed}
  {%
   \def\FrameCommand##1{##1\hspace{10pt}\color{red}{\vrule width 3pt}}%
   \MakeFramed {\advance\hsize-\width \FrameRestore}%
  }%
  {\endMakeFramed}
  {%
   \def\FrameCommand##1{##1\hspace{10pt}\color{yellow}{\vrule width 3pt}}%
   \MakeFramed {\advance\hsize-\width \FrameRestore}%
  }%
  {\endMakeFramed}
  {%
   \def\FrameCommand##1{##1\hspace{10pt}\color{green}{\vrule width 3pt}}%
   \MakeFramed {\advance\hsize-\width \FrameRestore}%
  }%
  {\endMakeFramed}

\title{$3.415$-Approximation for Coflow Scheduling \\ via Iterated Rounding}
\author{Lars Rohwedder\thanks{University of Southern Denmark, Odense, Denmark. Supported by Dutch Research Council (NWO) project “The Twilight Zone of Efficiency: Optimality of Quasi-Polynomial Time Algorithms”
[grant number OCEN.W.21.268]} \and Leander Schnaars\thanks{Technical University of Munich, Munich, Germany. Supported by the Deutsche Forschungsgemeinschaft (DFG, German Research Foundation) - GRK 2201/2 - Projektnummer 277991500}}
\date{}

\begin{document}
\maketitle

\begin{abstract}
\noindent We provide an algorithm giving a $\frac{140}{41} (<3.415)$-approximation for Coflow Scheduling and a $4.36$-approximation for Coflow Scheduling with release dates. This improves upon the best known $4$- and respectively $5$-approximations and addresses an open question posed by \cite{agarwal_sincronia_2018, fukunaga22} and others. We additionally show that in an asymptotic setting, the algorithm achieves a $(2+\epsilon)$-approximation, which is essentially optimal under $\mathbb{P} \neq \mathbb{NP}$. The improvements are achieved using a novel edge allocation scheme using iterated LP rounding together with a framework which enables establishing strong bounds for combinations of several edge allocation algorithms.
\end{abstract}

\section{Introduction}
Coflow Scheduling models the problem of data exchange between various nodes in a shared network. It has been proven indispensable in improving the performance of common data exchange and distributed computing frameworks such as MapReduce \cite{Dean2008}, Spark \cite{Spark2010}, and Hadoop \cite{Hadoop2010}. These routines form an integral part for large scale computations commonly found in applications such as  bioinformatics, deep learning, and large language models \cite{Guo2018,Mostafaeipour2020,Gupta2017}. The problem has enjoyed attention both from the theory community as well as the application side, with many works spanning the bridge between theory and practice.\\

Formally, a coflow instance is given by some bipartite set of vertices $V := U_1 \cup U_2$ and a set of coflows $E_1,\dotsc,E_n$, where each coflow $E_j$ is a subset of bipartite edges on $V$, possibly containing duplicates. Additionally, each coflow $E_j$ has some associated weight $\omega_j \in \mathbb{R}^+$. This models for example a set of input and output ports in a shared network, where each edge inside a coflow represents some data transmission requirement. During each discrete step in time, we are allowed to schedule a set of edges on the graph for which no vertex has more than one adjacent edge, so a matching. This represents the requirement that ports only send to and receive from one other port during each discrete step in time. A coflow finishes at time-step $t$ if all of its edges have been scheduled on or before time $t$, with at least one edge being scheduled during $t$. We call $C^*_j$ the finishing time of coflow $E_j$ and wish to minimize the weighted sum of completion times $\sum_{j \in [n]}\omega_j C^*_j$.\\

Coflow Scheduling was first introduced by Chowdhury et al. \cite{chowdhury_efficient_2014}, though the closely related problem of scheduling on network switches has been studied earlier under different names by various authors \cite{bonuccelli_incremental_1991,gupta_scheduling_1996}. Coflow Scheduling can be seen as an extension with a combinatorial structure of a problem called Concurrent Open Shop Scheduling (COSS) with preemption. In COSS, there is a set of machines and jobs, where every job has some demand on each machine which can be fulfilled concurrently, and jobs finish when they are completed on every machine. For COSS, several $2$-approximation algorithms are known \cite{garg_order_2007,leung_scheduling_2007, mastrolilli_minimizing_2010} and the problem is known to be $\mathbb{NP}$-hard to approximate within $2-\epsilon$, for any $\epsilon > 0$ \cite{Sachdeva2013}. This hardness result extends to Coflow Scheduling. On the side of approximation algorithms for Coflow Scheduling, there is still a gap to the lower bound. For the case of no release dates, multiple authors have given $4$-approximation algorithms \cite{agarwal_sincronia_2018, ahmadi_scheduling_2017, fukunaga22}, which extend to $5$-approximations in the case of release dates. Fukunaga \cite{fukunaga22} shows that in the case of release dates the integrality gap of the linear program used in the algorithm is at most $4$, though his proof is non-constructive. Several authors have claimed $(2+\epsilon)$-approximations, but all were later shown to be incorrect, see \cite{im2018tightapp, ahmadi_scheduling_2017} for discussion.
For the setting in which the simultaneously schedulable flows have to be independent sets of a matroid instead of matchings, Im et al. \cite{im19} provided an algorithm with a $(2+\epsilon)$-approximation guarantee.
Khuller et al. \cite{khuller_select_2019} showed a framework which provides guarantees in an online setting using offline approximation algorithms, leading to a $12$-approximation for online Coflow Scheduling. Extensions to general graphs \cite{jahanjou_asymptotically_2017} and so called path-based Coflow Scheduling \cite{eckl_minimization_2020} have been studied.\\

Whether the $4$- and respectively $5$-approximation can be improved has been a major open question raised by most previous works \cite{agarwal_sincronia_2018, fukunaga22}, especially in light of the $2-\epsilon$ lower bound. We address this question and show that both bounds can be beaten, with a tight result in an asymptotic setting.

\subsection{Our Contribution}
\noindent We present the first polynomial time algorithm which achieves a better than $4$-approximation for Coflow Scheduling without release dates and the first algorithm which achieves a better than $5$-approximation with release dates.
More specifically, we show the following theorems.

\begin{restatable}{theorem}{rsthmcoflowapprox}\label{thm:coflowapprox}
There is a polynomial time algorithm achieving a $\frac{140}{41} (<3.415)$-approximation for Coflow Scheduling without release dates.
\end{restatable}

\begin{restatable}{theorem}{rsthmcoflowapproxrelease}\label{thm:coflowapproxrelease}
There is a polynomial time algorithm achieving a $4.36$-approximation for Coflow Scheduling with release dates.
\end{restatable}

Using a more technical construction, the guarantee of Theorem \ref{thm:coflowapprox} can be slightly improved, see Section \ref{sec:app:approximprov} for details. We additionally prove that in a certain asymptotic setting, roughly when most coflows have large finishing times in any optimum solution, we can achieve a $(2+\epsilon)$-approximation, which is optimal. This result holds even in the case with release dates.

\begin{restatable}{theorem}{rsthmcoflowapproxlarge}\label{thm:coflowapproxlarge}
 For any $\epsilon > 0$, there exists an $\hat{\epsilon} > 0$ such that there is a $(2+\epsilon)$-approximation algorithm for all coflow instances $\mathcal{I}$ fulfilling
 \begin{equation*}
     \sum_{j \in [|\mathcal{I}|]}\omega_j \quad \le \quad \hat{\epsilon}\cdot \mathrm{OPT}(\mathcal{I}).
 \end{equation*}
\end{restatable}

\noindent Note that using a framework by Khuller et al. \cite{khuller_select_2019}, any improvement in the approximation ratio for Coflow Scheduling without release dates directly gives an improvement for the best known approximation for the online setting of Coflow Scheduling.
The following table provides an overview over the state of the art of approximation algorithms for various variants of Coflow Scheduling and our respective improvements.

\begin{table}[H]
\centering
\renewcommand{\arraystretch}{1.2}
\begin{tabular}{l@{\hskip 6\tabcolsep}l@{\hskip 6\tabcolsep}l}\toprule
      Case     & Best Known  & This work\\\midrule
No Release Dates    & 4\phantom{00} [*] & 3.415\\
Release Dates & 5\phantom{00} [*] & 4.36\\
Release Dates (integrality gap) & 4\phantom{00} \cite{fukunaga22} & $3.893$\\
Asymptotic + No Release Dates & 4\phantom{00} [*] & $2+\epsilon$\\
Asymptotic + Release Dates   & 5\phantom{00} [*] &  $2+\epsilon$\\
Online & 12\phantom{0} \cite{khuller_select_2019} & $11.415$\\\bottomrule
\end{tabular}
\caption{Best known previous approximations and our results. The sources marked [*] are \cite{ahmadi_scheduling_2017,agarwal_sincronia_2018,fukunaga22}.}
\label{tab:algperf}
\end{table}

The main technical contributions are a novel allocation and rounding scheme for the individual edges of each coflow, inspired by techniques used in proving the Beck-Fiala Theorem from discrepancy theory and a framework for establishing the approximation ratio for combinations of several such edge scheduling algorithms.\\
Our algorithm follows a two phase approach. This has been done either implicitly or explicitly in most approaches found in the literature. In the first phase, for each coflow and its associated flows, deadlines are determined through an LP based approach combined with a randomized rounding procedure. These deadlines are equipped with a special structure which we exploit in the algorithm. They specifically provide a $2$-approximation cost guarantee with respect to some optimum solution.
In the second phase, the goal is to find a valid allocation for the coflows, or more precisely for the individual edges of each coflow, to time-slots. In previous works, this was achieved by a simple greedy allocation procedure, which in the case without release dates achieves a deadline violation of at most a factor $2$ for each coflow, yielding a $4$-approximation. We use a combination of two algorithms, both the greedy allocation rule and a novel iterated rounding scheme. The greedy allocation performs well for small deadlines, but converges to a $2$-approximation for larger deadlines, while the rounding procedure can schedule large deadlines arbitrarily well, but has worse results for small ones. Note that these factors only capture the completion time delay and do not take into account the additional loss of factor $2$ from the deadline construction. By running both algorithms and picking the better result, we are able to improve upon the factor $4$.\\
To show the improved approximation ratio, we establish a general framework which can be used to bound the coflow scheduling approximation ratio for any collection of edge allocation algorithms.

\subsection{Organization}
The next chapter introduces important notation and discusses results from LP and graph theory which form an important part of several subroutines. Section \ref{sec:algframe} provides a complete overview of the entire algorithm and the most important theorems for Coflow Scheduling without release dates. Full proofs and additional details can be found in the following Section \ref{sec:schedflows}. Section \ref{sec:asymp} proves the $2+\epsilon$ guarantee in the asymptotic case. Appendix \ref{sec:app} provides additional details and discusses some further related theory, such as results regarding the structure and complexity of the used LP, extending the algorithm to release dates and high edge multiplicities, the improved LP integrality gap, and shows how to achieve a slight improvement over the approximation guarantee from Theorem \ref{thm:cbf}.

\section{Preliminaries}

For $n \in \mathbb{N}$ we define $[n] := \{1,\dotsc,n\}$. Graphs $G = (V,E)$ are defined by a set of vertices $V$ and edges $E \subseteq V \times V$. We slightly abuse notation and allow $E$ to contain multiple copies of the same edge. We use $\Delta(G)$ to denote the maximum degree of a graph. For some set of edges $E$, $\Delta(E)$ refers to the maximum degree of the canonical graph induced by this edge set. \\

A coflow instance is given by a collection of bipartite edge multi-sets $E_1,\dotsc,E_n$ on some set of vertices, together with weights $\omega_1,\dotsc,\omega_n \in \mathbb{R}_+$. We usually only refer to the edge sets and let the vertices of the underlying graph be implicitly described by them. We define $E := \cup_{j \in [n]}E_j$. In the case of release dates, for every coflow $j \in [n]$ there is a release date $r_j \in \mathbb{N}$. A release date $r_j$ means that edges from coflow $E_j$ can be scheduled the earliest in time slot $r_j+1$. We use the term flow to refer to a collection of identical edges within one coflow, so essentially an edge with its multiplicity. In the main body of this work we assume that edge multiplicities are encoded in this explicit way, meaning that multiplicities are represented by multiple copies. We discuss the case where the multiplicities are instead encoded as an integer in Section \ref{sec:app:coflowpseudopoly}. A valid schedule is a mapping of all edges to time slots, such that in every time slot the assigned edges form a matching. In the case of release dates, no edge from any coflow is allowed to be scheduled in a time slot smaller than or equal to the respective release date. The finishing time of a coflow is the latest time slot to which one of its edges is assigned. We call $C^*_j$ the finishing time of coflow $E_j$ and wish to minimize the weighted sum of completion times $\sum_{j \in [n]}\omega_jC^*_{j}$.\\

As we use an iterated LP rounding scheme, we give a brief overview of the most important relevant theory here. For more details see for example \cite{bansal:LIPIcs.FSTTCS.2014.1,schrijver_linprog}. Let $A \in \mathbb{R}^{m \times n}$ be a matrix and $b \in \mathbb{R}^m, c \in \mathbb{R}^n$ be vectors. We consider the polytope $\mathcal{P} = \{x\in [0,1]\ |\ Ax \le b\}$ and the associated LP $\min_{x \in \mathcal{P}} c^Tx$. From standard LP theory we know that there always exists an LP optimum solution at a vertex of $\mathcal{P}$ which can be found in strongly polynomial time, as long as all values in $A$ are polynomially bounded \cite{Tardos1986}. One key fact we use is that additionally if $m < n$, one can find such a vertex solution in which at least $n - m$ entries are in $\{0,1\}$. This theorem can be extended to work when every entry $x_i$ is constrained to some interval $[0,n_i]$, for $n_i \in \mathbb{N}$.\\

The Coflow Scheduling problem is closely related to several well studied graph and coloring problems. As we use some of these results directly and implicitly in our work, we briefly review them here. Given some set of edges $E$ on a bipartite graph $G=(V,E)$, the question whether they can be partitioned into some number of matchings $k$ is equivalent to asking whether there is a proper $k$-edge-coloring of $E$. Clearly at least $\Delta(G)$, i.e. the maximum degree of the graph, colors are needed. Kőnig's Theorem, and to an extent also Vizing's Theorem, give a strong result for bipartite graphs:

\begin{theorem}[\cite{König1916}]\label{thm:koenig}
Any bipartite graph $G$ can be properly edge-colored with $\Delta(G)$ colors.
\end{theorem}

As a set of matchings is equivalent to a set of scheduled flows in the coflow setting, the question whether a set of edges can be scheduled during some collection of time points is equivalently answered by this. Any set of flows for which the induced graph has maximum degree $d$ can be scheduled within $d$ time slots. This result is essential to our analysis, as it shows that degree bounds for some set of selected flows are sufficient to ensure their schedulability.\\
Note that the proof of Theorem \ref{thm:koenig} can be done in a constructive way, leading to a polynomial time algorithm producing such an allocation. This algorithm can be extended to work even for edges with possibly superpolynomial multiplicities, for details see \cite{zhen_2015}.

%%%%%%%%%%%%%%%%%%%%%%%%%%%%%%%%%%%%%%%%%%%%%%%%%%%%%%%%%%%%%%%%%%%%%%%
\section{Algorithmic Framework}\label{sec:algframe}
This section provides a complete overview of the algorithmic framework used to establish the improved approximation ratios. We focus on the case without release dates here, the necessary modifications for release dates are described in Section \ref{sec:app:frameworkrelease}.

\subsection{Coflow Deadlines}\label{sec:coflowdead}
Given some coflow instance, we aim to determine a deadline for each of its coflows. These deadlines might not necessarily be strict in the sense that constructed schedules have to adhere to them, but they are rather used to both guide edge allocation procedures and to then bound their resulting costs. Most existing coflow approximation algorithms use a similar strategy.\\
We take a structural approach, where we first define structure which we want our deadlines to obey and then describe how such deadlines can be found. We capture the structural constraints in the following LP. It has been implicitly used in analysis by \cite{im19,fukunaga22} and others. Let $C_1 \le C_2 \le \dotsc \le C_n$ be deadlines for the coflows and for easier notation define $C_0 := 0$.

\begin{gather*}\tag{LP I}\label{lp:main}
\begin{aligned}
\sum_{s \in [n]}x_{s,e} &&=& \quad 1 &&\forall e \in E & (I)\\
\sum_{e: v \in e} x_{s,e} &&\le& \quad C_s - C_{s-1}\quad &&\forall s \in [n],\forall v \in V & (II)\\
x_{s,e} &&=& \quad 0 &&\forall j \in [n], \forall e \in E_j, \forall s > j\quad & (III)\\
x_{s,e} &&\ge&\quad 0
\end{aligned}
\end{gather*}

Instead of enforcing the necessary constraints for Coflow Scheduling for each individual time slot, \ref{lp:main} groups the slots into blocks in between the coflow deadlines. The variable $x_{s,e}$ describes the assignment of edge $e$ to block $s$. Constraint $(I)$ ensures that every edge from every coflow is fully scheduled. The block between $C_{s-1}$ and $C_{s}$ has size $C_{s} - C_{s-1}$, so constraint $(II)$ ensures that in every such block for every vertex the amount of adjacent edges does not exceed the block size. Constraint $(III)$ forces edges to be zero for blocks after the respective deadline. The step from individual time slot degree bounds to block degree bounds is justified by Kőnig's Theorem (Theorem \ref{thm:koenig}), as it guarantees that any bipartite graph with some maximum degree $\Delta$ can be decomposed into $\Delta$ matchings. This is equivalent to saying that any set of bipartite edges $E$ for which the induced graph has maximum degree $\Delta$ can be scheduled in $\Delta$ time slots.\\

Assuming the deadlines $C_1 \le C_2 \le \dotsc \le C_n$ are set as the coflow finishing times from some optimal schedule for the underlying coflow instance, the edge assignment directly induces a feasible point inside \ref{lp:main}. For each edge $e$ which is scheduled in some time slot $t$ in the optimal schedule, set $x_{s,e} := 1$, for $s$ such that $t \in (C_{s-1},C_{s}]$. Set all other variables to $0$. As by definition of a valid solution, each edge is scheduled before its coflow's deadline, constraint $(I)$ is fulfilled and constraint $(III)$ cannot be violated. As the edges in each time slot form a matching, constraint $(II)$ is also fulfilled.\\

Conversely, an integral solution to the LP corresponds to a valid solution for Coflow Scheduling. However, in order to be able to solve the LP in polynomial time, we cannot enforce integrality. Hence the constraints only enforce that the variable assignment corresponds to a fractional matching. There are instances and deadlines for which \ref{lp:main} is feasible, but no feasible integral point and therefore also no feasible integral schedule exists. In fact, determining whether an integral points exists is an $\mathbb{NP}$-hard problem, for details see Section \ref{sec:app:lpintegrality} in the appendix.\\

Finding some set of deadlines for which \ref{lp:main} is feasible is easy, as one can simply choose large enough values to guarantee feasibility. However, for the purpose of constructing good approximation algorithms for Coflow Scheduling, we require that the deadlines fulfill some cost guarantees with respect to an optimal coflow schedule.\\
There is an LP based approach which returns deadlines for which \ref{lp:main} is feasible and certain strong guarantees hold. This technique has been used by \cite{im19,fukunaga22} and others. They use a randomized rounding scheme on another LP formulation to determine integral deadlines $C'_1,\dotsc, C'_n$. We slightly modify their algorithm and leave out the final step in which they round up the deadlines and obtain $C_1,\dotsc,C_n$. These deadlines are thus potentially fractional. By slightly modifying their proof, the following bound can be shown.

\begin{restatable}[\cite{im19}]{lemma}{rslemmatwoapproxstrong}\label{lemma:2approxstrong}
There is a polynomial time randomized algorithm determining deadlines $C_1,\dotsc,C_n$ for which \ref{lp:main} is feasible and for which the following cost bound holds: 
\begin{equation*}
\sum_{j \in [n]} \omega_j \mathbb{E}[C_j] \quad  \le \quad  2\cdot \mathrm{OPT} \ -\  \sum_{j \in [n]}\omega_j
\end{equation*}
\end{restatable}

More details about the procedure used by \cite{im19} to determine such deadlines can be found in the appendix in Section \ref{sec:app:coflowdeadline}. They only implicitly work with \ref{lp:main}, so we provide additional details on the connection. Note that the procedure can be de-randomized to obtain a fully deterministic algorithm.\\

The multiplicative factor of $2$ in Lemma \ref{lemma:2approxstrong} is optimal assuming $\mathbb{P} \neq \mathbb{NP}$. This follows from the factor $(2-\epsilon)$-approximation hardness of Concurrent Open Shop Scheduling, as Coflow Scheduling can be seen as a generalization of this problem \cite{Sachdeva2013}.

\subsection{Integral Edge Assignments with Guarantees}
Let $C_1,\dotsc,C_n$ be deadlines for which \ref{lp:main} is feasible and Lemma \ref{lemma:2approxstrong} holds. Using the result of the lemma, we immediately obtain that if we are able to find an allocation such that all edges from each coflow are scheduled by their respective deadlines, we have achieved a $2$-approximation for Coflow Scheduling. In the same way, if for some $\alpha \ge 1$ we are able to schedule each coflow $j$ by time $\alpha \cdot C_j$, we obtain a $2 \cdot \alpha$ approximation algorithm.\\

We analyze two edge allocation algorithms which provide different guarantees for the finishing times of the coflows. The first algorithm $\mathrm{Greedy}$ is a simple greedy allocation scheme. This procedure was used by previous authors to derive $4$-approximation algorithms for Coflow Scheduling. Let $\mathrm{Greedy}(C_j)$ denote the finishing time of coflow $E_j$ in the schedule produced by $\mathrm{Greedy}$.
\begin{restatable}{lemma}{rslemmagreedy}\label{lemma:greedy}
For given deadlines $C_1,\dotsc,C_n$ for which \ref{lp:main} is feasible there is an algorithm $\mathrm{Greedy}$ returning a valid coflow schedule such that the following holds.
\begin{equation*}
\mathrm{Greedy}(C_j) \quad \le \quad 2C_j-1
\end{equation*}
\end{restatable}

The second algorithm $\mathrm{CBF}^\tau$ is a novel allocation scheme using a form of iterated rounding inspired by the Beck-Fiala Theorem from discrepancy theory \cite{Beck1981}. The algorithm is parameterized by $\tau \in \mathbb{N}_{\ge 2}$ and allocates the coflow deadlines to blocks, where each block's size is some integer multiple of $\tau$. An edge assignment is then determined which only slightly violates the size of each block. This leads to the following completion time guarantees.

\begin{restatable}{theorem}{rsthmcbf}\label{thm:cbf}
For given deadlines $C_1,\dotsc,C_n$ for which \ref{lp:main} is feasible, weights $\omega_1,\dotsc,\omega_n$, and a parameter $\tau \in \mathbb{N}_{\ge 2}$, there is an algorithm $\mathrm{CBF}^\tau$ returning a valid coflow schedule such that the following holds.
\begin{equation*}
   \sum_{j \in [n]} \omega_j \cdot \mathrm{CBF}^\tau(C_j) \quad \le \quad \sum_{j \in [n]}\omega_j \left(\frac{\tau+2}{\tau}C_j + \frac{\tau}{2} + 2.5 - \frac{2}{\tau}\right)
\end{equation*}
\end{restatable}

Note that the approximation guarantees of both algorithms are quite different. $\mathrm{Greedy}$ achieves rather strong approximation for small deadlines, while for large deadlines, through an appropriate choice of $\tau$, $\mathrm{CBF}^\tau$ gives good guarantees. In fact, under certain assumptions on the provided deadlines, $\mathrm{CBF}^\tau$ can achieve approximations arbitrarily close to the optimum of $1$, see Section \ref{sec:asymp}.\\

Our coflow algorithm aims to achieve guarantees for both small and large deadlines by combining both algorithms in some way. The procedure is straightfoward. It obtains the deadlines $C_1,..,C_n$ through the procedure explained in Section \ref{sec:coflowdead} and then runs both $\mathrm{Greedy}$ and $\mathrm{CBF}^6$ independently, returning the schedule with lower cost. Analyzing the cost of the returned solution requires some care, as we need a uniform bound over any possible input instance.

\subsection{Combining Algorithmic Guarantees}\label{sec:combalgguarantees}
By the definition of the coflow algorithm, for each possible instance $\mathcal{I}$, its cost $C_{\mathrm{ALG}}(\mathcal{I})$ is given as the minimum of the costs $C_G(\mathcal{I})$ and $C_{\mathrm{CBF}}(\mathcal{I})$ of the $\mathrm{Greedy}$ and respectively $\mathrm{CBF}^6$ algorithm.\\
We show a general proof framework for such algorithms which provide deadline guarantees, which gives sharp bounds for taking minimums over several algorithms' costs. The derivation is not difficult, but the framework offers a surprisingly simple and strong method to establish bounds for large classes of algorithms. It only requires bounds for the delay guarantees of the algorithms, which are usually relatively simple to establish.\\
For this purpose, let $C_1,\dotsc,C_n$ be deadlines for which \ref{lp:main} is feasible and let $\mathrm{ALG}_1,\dotsc,\mathrm{ALG}_k$ be algorithms producing valid coflow schedules from such deadlines, with $\mathrm{ALG}_i(C_j)$ being the finishing time of coflow $E_j$ in the schedule produced by $\mathrm{ALG}_i$. For $j \in [k]$, let $f_j$ be a function capturing a bound on the maximum weighted deadline delay of $\mathrm{ALG}_j$. This means that $f_j$ is such that $\sum_{j \in [n]}\omega_j \cdot \mathrm{ALG}_j(C_j) \le \sum_{j \in [n]}\omega_j \cdot f_j(C_j)$. Such functions might stem from bounds on individual deadlines like in the case of $\mathrm{Greedy}$, but can also come from bounds which are already given as a weighted sum over all deadlines like for $\mathrm{CBF}^\tau$.

\begin{restatable}{lemma}{rsthmappbound}
\label{lemma:appbound}
Let $\lambda_1,\dotsc,\lambda_k \ge 0$ with $\sum_{i \in [k]}\lambda_i = 1$ and $\alpha \in \mathbb{R}^+$. If for all $x \ge 1$
\begin{equation*}
    \sum_{i \in [k]}\lambda_i f_i(x) \quad \le \quad \alpha(x+1),
\end{equation*}
then for all coflow instances $\mathcal{I}$:
\begin{equation*}
    C_{\mathrm{ALG}}(\mathcal{I}) \Hquad=\Hquad \min\{C_{\mathrm{ALG}_1}(\mathcal{I}),\dotsc,C_{\mathrm{ALG}_k}(\mathcal{I})\} \Hquad\le\Hquad 2\alpha \cdot \mathrm{OPT}(\mathcal{I})
\end{equation*}
\end{restatable}

\begin{proof}
Let $\lambda_1,\dotsc,\lambda_k \ge 0$ be fixed constants with $\sum_{j \in [k]}\lambda_j = 1$. Define $g(x) := \lambda_1 f_1(x) + \lambda_2 f_2(x) + \cdots + \lambda_kf_k(x)$. We define a randomized algorithm $\mathrm{RALG}$ which for all $j \in [k]$ runs algorithm $\mathrm{ALG}_j$ with probability $\lambda_j$. With these definitions, for the cost $C_{\mathrm{ALG}}$ of the combined algorithm $\mathrm{ALG}$ we obtain:
\begin{equation*}
C_{\mathrm{ALG}} \Hquad=\Hquad \min \{C_{\mathrm{ALG}_1},\dotsc,C_{\mathrm{ALG}_k}\}
\Hquad\le\Hquad \sum_{j \in [k]}\lambda_j \cdot C_{\mathrm{ALG}_j}\Hquad =\Hquad \mathbb{E}[C_{\mathrm{RALG}}]
\end{equation*}
For the expected cost of $\mathrm{RALG}$ we have
\begin{align*}
\mathbb{E}[C_{\mathrm{RALG}}]&\Hquad=\Hquad \mathbb{E}\Big[\sum_{j \in [n]}\omega_j \mathrm{RALG}(C_j)\Big] \Hquad=\Hquad \sum_{j \in [n]}\omega_j \mathbb{E}[\mathrm{RALG}(C_j)]\\
~&\phantom{------}\Hquad\le\Hquad \sum_{j \in [n]}\omega_j \left(\lambda_1f_1(C_j) + \cdots + \lambda_kf_k(C_j)\right) \Hquad=\Hquad \sum_{j \in [n]}\omega_j g(C_j).\\
\end{align*}
So by establishing suitable bounds for $g(x)$, we can show approximation bounds for $\mathrm{ALG}$. Assume that there exists some $\alpha \in \mathbb{R}_+$ such that $g(x) \le \alpha(x+1)$. Then we can further bound
\begin{equation*}
 \sum_{j \in [n]}\omega_j g(C_j) \Hquad\le\Hquad \sum_{j \in [n]}\omega_j (\alpha(C_j+1)) \Hquad\le\Hquad \alpha \cdot \sum_{j \in [n]}\omega_jC_j \Hquad + \Hquad \alpha \cdot \sum_{j \in [n]}\omega_j.
\end{equation*}
By using the deadlines $C_1,\dotsc,C_n$ provided by Lemma \ref{lemma:2approxstrong} for which $\sum_{j \in [n]}\omega_j C_j \le 2 \cdot \mathrm{OPT} - \sum_{j \in [n]}\omega_j$ holds, this yields the desired bound:
\begin{equation*}
    C_{\mathrm{ALG}} \quad\le\quad 2\alpha \cdot \mathrm{OPT}\qedhere
\end{equation*}
\end{proof}

\subsection{Main Theorem}
Using the previous lemmata and theorems, we prove Theorem \ref{thm:coflowapprox}. The proof follows by application of the framework from Theorem \ref{lemma:appbound} to selected edge allocation algorithms.

\rsthmcoflowapprox*

\begin{proof}
For completeness, we restate the algorithm which has implicitly been described earlier. Given some coflow instance $\mathcal{I}$, we first determine deadlines using the procedure described in Section \ref{sec:coflowdead}. We then apply the two edge allocation algorithms $\mathrm{Greedy}$ and $\mathrm{CBF}^6$ to the deadlines to obtain two feasible coflow schedules and return the schedule with lower total cost.\\

Calling this algorithm $\mathrm{ALG}$, its cost $C_{\mathrm{ALG}}(\mathcal{I})$ is thus given as $\min\{C_{\mathrm{Greedy}}(\mathcal{I}), C_{\mathrm{CBF}^6}(\mathcal{I})\}$. We aim to use Theorem \ref{lemma:appbound} to bound the approximation ratio of $\mathrm{ALG}$. For this purpose, let $f_G$ be a function capturing an upper bound on the deadline delay of $\mathrm{Greedy}$ in the sense required by Theorem \ref{lemma:appbound} and respectively $f_{\mathrm{CBF}}$ for $\mathrm{CBF}^6$. From Lemma \ref{lemma:greedy} and Theorem \ref{thm:cbf} we obtain that
\begin{equation*}
    f_G(x) \Hquad=\Hquad  2x-1 \qquad \text{ and } \qquad f_{\mathrm{CBF}}(x) \Hquad=\Hquad \frac{4}{3}x + \frac{31}{6}   
\end{equation*}
holds. Let $\lambda_1 := 23/41$ and $\lambda_2 := 18/41$. This yields:
\begin{equation*}
    \lambda_1 f_G(x) + \lambda_2f_{\mathrm{CBF}}(x) \Hquad=\Hquad (2 \cdot \tfrac{23}{41} + \tfrac{4}{3}\cdot \tfrac{18}{41})x + (\tfrac{31}{6} \cdot \tfrac{18}{41} - \tfrac{23}{41}) \Hquad=\Hquad \tfrac{70}{41}(x+1)
\end{equation*}
So the requirements of Theorem \ref{lemma:appbound} are fulfilled with $\alpha = \tfrac{70}{41}$, which implies that $\mathrm{ALG}$ is a $2 \cdot \tfrac{70}{41} = \tfrac{140}{41} < 3.415$-approximation algorithm for Coflow Scheduling without release dates.
\end{proof}

%%%%%%%%%%%%%%  SECTION  %%%%%%%%%%%%%%%%%%%%%
\section{Integral Edge Assignments with Guarantees}\label{sec:schedflows}

In this section we introduce and analyze algorithms which allocate edges of coflows to time-slots. They work on coflow deadlines fulfilling certain structural properties and their goal is to provide feasible schedules together with guarantees on the average delay each coflow experiences. Similar strategies are also used in most of the previous $4$-approximations for Coflow Scheduling. We introduce the algorithms $\mathrm{Greedy}$ and $\mathrm{CBF}^\tau$ and show their guarantees in Lemma \ref{lemma:greedy} and Theorem \ref{thm:cbf}.

\subsection{Greedy Scheduling}
We start by introducing and analyzing a greedy allocation algorithm $\mathrm{Greedy}$, which is one of the edge allocation procedures used in previous works to achieve a $4$-approximation. Let $C_1 \le C_2 \le \cdots \le C_n$ be deadlines for which \ref{lp:main} is feasible. $\mathrm{Greedy}$ schedules all coflows consecutively, starting with $E_1$ up to $E_n$. Each edge is simply scheduled in a work-conserving way, meaning that it is scheduled in the earliest possible time-slot in which both its vertices are free. By doing this for all edges, a schedule is obtained. For $j \in [n]$, let $\mathrm{Greedy}(C_j)$ be the finishing time of coflow $E_j$ in the schedule obtained from this procedure.

\rslemmagreedy*

\begin{proof}
Consider some fixed $j\in [n]$ and let $e = (u,v) \in E_j$ be an edge on which coflow $j$ finishes. As \ref{lp:main} is feasible for the deadlines, for both $u$ and $v$ at most $C_j - 1$ flow which contains one of these vertices from earlier coflows can exist. This implies that at most $\lfloor C_j-1\rfloor$ edges in $(\cup_{i \le j}E_i)\setminus \{e\}$ can be adjacent to each of $u$ and $v$. So these edges can block at most $2(\lfloor C_j - 1 \rfloor) \le 2C_j -2$ time slots, which implies that $\mathrm{Greedy}$ schedules $e$ the latest in slot $2C_j - 1$.
\end{proof}

$\mathrm{Greedy}$ provides a strict deadline guarantee for each coflow, meaning that in the schedule produced every coflow finishes the latest at the provided bound. This also implies that the same guarantee holds when taking weighted sums over the finishing times. Note that $2C_j - 1 = (2 - \frac{1}{C_j})C_j$, so for small $C_j$ this gives a tangible improvement over the factor $2$.

%%%%%%%%%%%%%%%%%%%%%%%%%%%%%%%%%%%%%%%%%%%%%%%%

\subsection{Iterated Rounding using Beck-Fiala} \label{sec:beckfiala}
This section gives a proof of Theorem \ref{thm:cbf}. We start by providing a full description of the algorithm, then we give a preliminary analysis and subsequently strengthen the guarantee through further refinements.
\subsubsection*{Procedure Idea}

Given some deadlines $C_1 \le C_2 \le \dotsc \le C_n$ for which \ref{lp:main} is feasible, the core idea is to round these deadlines to the next integer multiple of some parameter $\tau \in \mathbb{N}_{\ge 2}$ and to then form blocks between consecutive rounded deadlines. With these blocks and associated deadlines, we show that it is possible to allocate all edges to blocks while only violating the block size by a small additive constant. Given such an allocation, using the guarantee provided by Kőnig's Theorem (Theorem \ref{thm:koenig}) there exists a feasible schedule containing all assigned edges within the maximum vertex load of each block. Through the rounding and the increase in blocks' sizes the finishing times of the coflows are delayed with respect to their deadlines. We are however able to show strong bounds on this delay.\\
We call this algorithm $\mathrm{CBF}^\tau$ due to its close association with the proof of the Beck-Fiala Theorem \cite{Beck1981}.

\subsubsection*{Edge-to-Block Allocation LP}
We use a rounding technique inspired by the proof of the Beck-Fiala Theorem. Let $\tau \in \mathbb{N}_{\ge 2}$ be a fixed constant. For $j \in [n]$ let $\bar{C}_j$ be the deadline $C_j$ rounded up to the next integer multiple of $\tau$. The blocks' sizes are defined by the distance between two non-equal consecutive deadlines. So for $\bar{C}_j \neq \bar{C}_{j-1}$, block $j$ has size $\bar{C}_j - \bar{C}_{j-1}$. We define the following LP, which models the allocation of coflow edges to blocks. One can assume without loss of generality that all rounded deadlines are distinct and that there are $n$ of them, as coflows whose rounded deadlines are equal can be joined in this step.

\begin{gather*}\tag{LP CBF}\label{lp:cbf}
\begin{aligned}
\sum_{b \in [n]}x_{e,b} \quad &= \quad 1 &&\forall e \in E &\qquad(I)\\
\sum_{e: v \in e} x_{e,b} \quad&\le \quad \bar{C}_b - \bar{C}_{b-1} &&\forall v \in V,\forall b \in [n]&\qquad(II)\\
x_{e,b} \quad&= \quad 0\quad &&\forall j \in [n],\forall e \in E_j,\forall b > j\\
x_{e,b} \quad&\ge\quad 0
\end{aligned}
\end{gather*}

The structure of \ref{lp:cbf} is identical to \ref{lp:main}, though the special form of the rounded deadlines induces some additional properties. By definition $\bar{C}_j \ge C_j$ for all $j \in [n]$. Additionally, with $\bar{C}_0 = 0$, for each $b \in [n]$ we have $\bar{C}_b - \bar{C}_{b-1} = k_b \cdot  \tau$, for some $k_b \in \mathbb{N}$.
We identify each edge and vertex in the coflow instance with the respective variable in the LP and use both terms interchangeably.\\

We claim that feasiblity of  \ref{lp:main} for $C_1,\dotsc,C_n$ directly implies feasibility of \ref{lp:cbf} for $\bar{C}_1,\dotsc,\bar{C}_n$. For \ref{lp:main}, increasing the value of any deadline without violating the total order can only increase the feasible region, as the equal zero constraints get less restrictive and any possible excess assignment can be shifted between the two adjacent blocks whose size changes. Therefore, as all deadlines can only increase, \ref{lp:cbf} also has to be feasible.

\subsubsection*{LP Rounding}
We describe a procedure which finds an integral edge-to-block assignment violating the block size constraint $(II)$ by a constant amount. To achieve this, we start with an initial solution to the LP and then successively refine and resolve the LP, until we have obtained an integral solution fulfilling certain strong properties.\\

From now on we assume that we started with deadlines $C_1,\dotsc,C_n$ for which \ref{lp:main} is feasible, so we know that for $\bar{C}_1,\dotsc,\bar{C}_n$ \ref{lp:cbf} is feasible. After obtaining an initial solution to \ref{lp:cbf}, we take two steps. We first fix all integral edges and remove their respective variables from the LP and modify the right hand side of $(II)$ accordingly and then we delete all constraints from $(II)$ with at most $k-1$ fractional variables remaining, for some fixed number $k \in \mathbb{N}$. This corresponds to dropping the degree constraint on the respective vertex if at most $k-1$ adjacent edges are still fractional. Let $\mathcal{E}_{b}$ be the set of all fractional edges contained in block $b$ and let $\mathcal{V}_{b}$ be the vertices in block $b$ with at least $k$ fractional adjacent edges. Let $S_{v,b}$ be the set of already fixed edges adjacent to $v$ in block $b$.
This gives rise to the following resulting LP:

\begin{gather*}
\begin{aligned}
\sum_{b \in [n]}x_{e,b} \quad &= \quad 1 &&\forall e \in \bigcup_{b \in [n]}\  \mathcal{E}_b& \quad(I)\\
\sum_{e \in \mathcal{E}_b: v \in e} x_{e,b} \quad&\le \quad k_b \cdot \tau - |S_{v,b}|\quad &&\forall b \in [n], v \in \mathcal{V}_b &\quad(II)\\
x_{e,b}\quad&\ge\quad 0 && \forall b \in [n], e \in \mathcal{E}_b
\end{aligned}
\end{gather*}

If we can show that this LP always contains strictly more variables than it contains constraints in $(I)$ and $(II)$, by considering a basic feasible solution, we obtain at least one more integral variable, so repeating the fixing variables and removing constraints step leads to at least one more fixed variable. Therefore in a polynomial number of steps the procedure must terminate and we obtain an integral solution. The step in which we drop constraints means that this integral solution is most likely not feasible for the original LP, but we later show that the amount of violation cannot be very large, which yields the desired approximation behaviour.

\subsubsection*{Constraints and Variables}
We want to show that the number of constraints is strictly smaller than the number of variables. The total number of variables in the LP is equal to $\sum_{b \in [n]}|\mathcal{E}_b|$. The number of constraints in $(I) \cup (II)$ is equal to $|\bigcup_{b \in [n]} \mathcal{E}_b| + \sum_{b \in [n]}|\mathcal{V}_b|$.
We show two bounds which enable us to establish the desired inequality.
\begin{lemma}
    For all $b \in [n]:$ $\quad|\mathcal{V}_b| \quad \le \quad \frac{2}{k}|\mathcal{E}_b|$
\end{lemma}
\begin{proof}
As by definition each vertex in $\mathcal{V}_b$ has at least $k$ fractional adjacent edges, we obtain the following. 
\begin{equation*}
|\mathcal{V}_b| \cdot k \quad \le \quad |\{(v,e) \ |\ \forall v \in \mathcal{V}_b, \forall e \in \delta_E(v) \cap \mathcal{E}_b\}|
\end{equation*}
Each edge contains exactly two vertices, so we additionally have the following inequality:
\begin{equation*}
|\{(v,e)\ |\ \forall v \in \mathcal{V}_b, \forall e \in \delta_E(v) \cap \mathcal{E}_b\}|\quad\le\quad 2 \cdot |\mathcal{E}_b|
\end{equation*}
Combining the two inequalities gives the result.
\end{proof}

\begin{lemma}
For all $b \in [n]:$ $\quad|\bigcup_{b \in [n]}\mathcal{E}_b|\quad \le \quad \frac{1}{2}\sum_{b \in [n]}|\mathcal{E}_b|$
\end{lemma}
\begin{proof}
For arbitrary sets, the bound is only true without the factor $\tfrac{1}{2}$ on the right hand side. Equality is reached exactly when all elements are unique. In our case, whenever there is a fractional edge, due to constraint (I), at least one other variable associated to this edge in another block has to be fractional as well. Hence, these contribute at least twice to the right hand side and only once to the left hand side, which gives the inequality. 
\end{proof}

\noindent Combining the two lemmata, we obtain:

\begin{equation*}
    |\mathrm{Cons}| \Hquad=\Hquad |\bigcup_{b \in [n]} \mathcal{E}_b| + \sum_{b \in [n]}|\mathcal{V}_b|\Hquad \le\Hquad \frac{1}{2}\sum_{b \in [n]}|\mathcal{E}_b| + \frac{2}{k}\sum_{b \in [n]}|\mathcal{E}_b| \Hquad=\Hquad \left(\frac{1}{2} + \frac{2}{k}\right)|\mathrm{Vars}|
\end{equation*}

So for all $k > 4$ a strict inequality follows. For $k=4$ we have the inequality $|\mathrm{Cons}| \le |\mathrm{Vars}|$. We can however still achieve a strict inequality for this case by slightly modifying the LP. In its current form, the LP is given without an objective function. We can thus remove one constraint from $(II)$ and shift it to the objective function instead. This reduces the number of constraints by one without changing the number of variables. If $b, v \in \mathcal{V}_b$ are the parameters corresponding to the chosen inequality, the added objective function is $\min \sum_{e\in\mathcal{E}_b: v \in e}x_{e,b}$. From the minimization objective it follows that feasible optimal points of the modified LP are feasible for the original LP, as the removed constraint cannot be violated.

\subsubsection*{Delay Bound}
\noindent Looking at the integral assignment, we can show that the violation of constraints of \ref{lp:cbf} is small. Note that the following statement only requires integrality of the deadlines and not the special structure of the rounded deadlines.
\begin{lemma}\label{lemma:cbflpviolation}
Given integral deadlines $C_1,\dotsc,C_n$ for which \ref{lp:cbf} is feasible and a parameter $\tau \in \mathbb{N}_{\ge 2}$, in polynomial time we can find an integral point such that:
\begin{itemize}
    \item[a)] All constraints $(II)$ in \ref{lp:cbf} are exceeded by at most $2$.
    \item[b)] All other constraints in \ref{lp:cbf} are fulfilled.
\end{itemize}
\end{lemma}
\begin{proof}
At all times during the procedure, the intermediate LP solutions fulfill the constraints not in $(II)$, so b) follows.
Remember that we never change variables once they are integral and that we only remove constraints from $(II)$ if the number of fractional variables in the sum is at most $k-1$. This shows that each constraint in $(II)$ can only be violated by an additive term of $k-1$. In fact, as the fractional variables have a strictly positive sum, the sum over the remaining integral variables at the time of removal can be at most $\tau-1$, which implies that the violation is at most $k-2$. For the choice $k=4$, a) follows.\\
\end{proof}

\noindent Using this result, we can show an upper bound on the maximum delay for each deadline when applying $\mathrm{CBF}^\tau$. In total, we obtain the following lemma, which at this point is slightly weaker than required for Theorem \ref{thm:cbf}, as there is an additive constant of $\tau+2$ instead of $\tau/2 + 2.5 - \frac{2}{\tau}$. Nevertheless, the theorem in this form would already be sufficient to gain a significant improvement over the factor $4$-approximation. Like in the case of $\mathrm{Greedy}$, $\mathrm{CBF}^\tau(C_j)$ is the finishing time of coflow $E_j$ in the schedule created by $\mathrm{CBF}^\tau$. 
\begin{lemma}\label{lemma:cbfweak}
For given deadlines $C_1,\dotsc,C_n$ for which \ref{lp:main} is feasible and a parameter $\tau \in \mathbb{N}_{\ge 2}$, there is an algorithm $\mathrm{CBF}^\tau$ returning a valid coflow schedule such that the following holds for all $j \in [n]$.
\begin{equation*}
    \mathrm{CBF}^\tau(C_j) \quad \le \quad \frac{\tau+2}{\tau}C_j + \tau + 2
\end{equation*}
\end{lemma}
\begin{proof}
The algorithm is given by setting the deadlines to the next integer multiple of $\tau$ and then doing the iterated rounding procedure. Given the edge to block assignments, a valid schedule can be obtained using Kőnigs Theorem (Theorem \ref{thm:koenig}).\\
Consider some fixed $C_j$ and let $k \in \mathbb{N}$ and $a \in [0,\tau)$ such that $C_j = k \cdot \tau + a$. Assume for now that $a > 0$, then $\bar{C}_j = (k+1) \tau$. By definition, the deadline $\bar{C}_j$ forms the $j$-th block. From Lemma \ref{lemma:cbflpviolation} it follows that each block's size increases by at most $2$, so the latest possible time at which coflow $j$ finishes is $(k+1) \tau + 2j$. As each block has size at least $\tau$, we have $j \le (k+1)$, so $(k+1)\tau + 2j \le (k+1)(\tau+2)$. Bounding this yields:
\begin{equation*}
    (k+1)(\tau+2) \Hquad\le\Hquad (\tau+2)k + \frac{\tau+2}{\tau}\cdot a + \tau + 2 \Hquad=\Hquad \frac{\tau+2}{\tau}C_j + \tau + 2
\end{equation*}
For $a = 0$, the argument simplifies. One has $\bar{C}_j = k\cdot\tau$ and hence a finishing time upper bound of $k(\tau+2)$. In this case a stronger bound of $\frac{\tau+2}{\tau}C_j$ follows.
\end{proof}

\subsubsection*{Reducing the Additive Constant}
The additive constant $\tau$ in Lemma \ref{lemma:cbfweak} assumes the worst case for each deadline, meaning that every deadline gets shifted from the very start to the very end of a block. We show that an averaging argument can be used to reduce the average amount of shift to $\tau/2 + \frac{1}{2} - \frac{2}{\tau}$. This requires that we show the bound across the weighted sum over all deadlines, unlike the previous proofs which established hard upper bounds for each individual deadline.\\

For $\lambda \in \mathbb{N}$, we consider a variant $\mathrm{CBF}^\tau_\lambda$ of $\mathrm{CBF}^\tau$ where an additional first block of fixed size $\lambda$ is inserted. This is equivalent to rounding the deadlines to the next larger term in the sequence $\{\lambda+i\cdot\tau\}_{i \in \mathbb{N}}$. Note that by simple modification of the arguments, the feasibility statements for \ref{lp:cbf} still apply. Lemma \ref{lemma:cbflpviolation} is thus also applicable.\\

This change to the deadline rounding step can change the finishing time of deadlines in our procedure. On the one hand, the delay of some deadlines might increase, as the last time slot of their respective blocks gets increased. On the other hand, the delay of some deadlines might decrease, as they now get included in an earlier block. We show the following.

\rsthmcbf*
\begin{proof}
The algorithm tries all $\lambda \in \{0,2,\dotsc,\tau-1,\tau+1\}$ and returns the solution with lowest total cost. We can upper bound this cost by instead considering a uniformly random $\lambda \in \{0,2,\dotsc,\tau-1,\tau+1\}$ and calculating the expected cost.
Consider some fixed $C_j$ and let $k \in \mathbb{N}, b \in \{0,1,\dotsc,\tau-1\}, a \in [0,1)$ such that $C_j = k\cdot \tau + b + a$. We assume for now that both $a\neq 0$ and $b \neq 0$. We write $\bar{C}_j(\lambda)$ to denote the smallest term in the sequence $\{\lambda+i\cdot\tau\}_{i \in \mathbb{N}}$ which is greater than or equal to $C_j$. \\
In case $\lambda = 0$, we have $\bar{C}_j(\lambda) = (k+1)\tau$ and by the same arguments as used in the proof of Lemma \ref{lemma:cbfweak} we obtain
\begin{equation*}
    \mathrm{CBF}^\tau_\lambda(C_j) \Hquad\le\Hquad (k+1)(\tau+2) \Hquad = \Hquad \frac{\tau+2}{\tau}C_j + \tau + 2 - \frac{\tau+2}{\tau}(a+b). 
\end{equation*}
In the case $\lambda \in \{2,\dotsc,b\}$, we have $\bar{C}_j(\lambda) = (k+1)\tau + \lambda$ and the index of $C_j$'s block increases by one. This gives:
\begin{equation*}
    \mathrm{CBF}^\tau_\lambda(C_j) \Hquad\le\Hquad (k+1)(\tau+2) + (\lambda+2) \Hquad=\Hquad \frac{\tau+2}{\tau}C_j + \tau + 2 - \frac{\tau+2}{\tau}(a+b) + (\lambda + 2) 
\end{equation*}
In the case $\lambda \in \{b+1,\dotsc, \tau-1\}$, we have $\bar{C}_j(\lambda) = k\tau + \lambda$ and the block index stays the same, thus:
\begin{equation*}
    \mathrm{CBF}^\tau_\lambda(C_j) \Hquad\le\Hquad k(\tau+2) + \lambda + 2 \Hquad=\Hquad \frac{\tau+2}{\tau}C_j + \tau + 2 - \frac{\tau+2}{\tau}(a+b) - (\tau -\lambda) 
\end{equation*}
In the case $\lambda = \tau+1$, as we assumed $a,b \neq 0$, we have $\bar{C}_j(\lambda) = (k+1)\tau + 1$ and the block index stays the same, so we have:
\begin{equation*}
    \mathrm{CBF}^\tau_\lambda(C_j) \Hquad\le\Hquad (k+1)(\tau+2) + 1 \Hquad=\Hquad \frac{\tau+2}{\tau}C_j + \tau + 2 - \frac{\tau+2}{\tau}(a+b) + 1
\end{equation*}
Every right hand side contains the same term $\hat{C}_j := \frac{\tau+2}{\tau}C_j + \tau + 2 - \frac{\tau+2}{\tau}(a+b)$, which is independent of $\lambda$, so we define a random variable $D_j(\lambda) := \mathrm{CBF}^\tau_\lambda(C_j) - \hat{C}_j$ which captures the respective remaining terms. Define $L_1 = \{2,\dotsc, b\}$ and $L_2 = \{b+1,\dotsc, \tau-1\}$. Then we have
\begin{align*}
\mathbb{E}[D_j] &\Hquad\le\Hquad \frac{b-1}{\tau}\cdot\mathbb{E}[2+\lambda\ |\ \lambda \in L_1] \Hquad - \Hquad \frac{\tau-1-b}{\tau}\cdot\mathbb{E}[\tau - \lambda\ |\ \lambda \in L_2] \Hquad+\Hquad \frac{1}{\tau}\cdot\mathbb{E}[1\ |\ \lambda \in \{\tau+1\}]\\
~   &\Hquad=\Hquad \frac{b-1}{\tau}\left(\frac{1}{b-1}\sum_{\lambda \in L_1} 2+\lambda\right) \Hquad-\Hquad \frac{\tau-1-b}{\tau}\left(\frac{1}{\tau-1-b}\sum_{\lambda \in L_2}\tau-\lambda\right) \Hquad + \Hquad \frac{1}{\tau}\\
    &\Hquad=\Hquad \frac{1}{\tau} \left(2(b-1) + \frac{b(b+1)}{2} - 1 - \frac{1}{2}(\tau-b)(\tau-1-b)\right) \Hquad + \Hquad \frac{1}{\tau}\\
    &\Hquad=\Hquad \frac{2b}{\tau} - \frac{\tau}{2} - \frac{2}{\tau} + \frac{1}{2} + b.
\end{align*}
So overall, we obtain
\begin{align*}
    \mathbb{E}[\mathrm{CBF}^\tau_\lambda(C_j)] &\Hquad\le\Hquad \hat{C}_j + \mathbb{E}[D_j]\\
    ~&\Hquad=\Hquad \frac{\tau+2}{\tau}C_j + \tau + 2 - \frac{\tau+2}{\tau}(a+b) + \left(\frac{2b}{\tau} - \frac{\tau}{2} - \frac{2}{\tau} + \frac{1}{2} + b\right)\\
    ~&\Hquad=\Hquad \frac{\tau+2}{\tau}C_j + \frac{\tau}{2} + 2.5 - \frac{2}{\tau} - (1+\tfrac{2}{\tau})a.
\end{align*}
As $a \in (0,1)$, the result follows. For $a=0$ or $b=0$, by checking all the cases one obtains that in each case, the bound stays the same or improves, so the theorem follows.\endproof
\end{proof}

%%%%%%%%%%%%%%  SECTION  %%%%%%%%%%%%%%%%%%%%%
\section{Asymptotic $2+\epsilon$ Approximation}\label{sec:asymp}
Theorem \ref{thm:coflowapprox} establishes that there is an algorithm returning a $3.415$-approximation for any given coflow input instance. In this section we show a stronger, asymptotically optimal, approximation for input instances with a certain structure. This result does not depend on the approximation framework but rather follows directly from the bounds established in Lemma \ref{lemma:cbfrelapprox}. Note that to show $(2-\epsilon)$-approximation hardness in \cite{Sachdeva2013}, they construct a sequence of instances for which the ratio between the sum over all weights and the optimum grows arbitrarily large, which shows that the asymptotic result in Theorem \ref{thm:coflowapproxlarge} is essentially optimal.

\rsthmcoflowapproxlarge*

\begin{proof}
For a given instance $\mathcal{I}$, using Lemma \ref{lemma:2approxstrong} we can obtain deadlines $C_1,\dotsc,C_n$ feasible for \ref{lp:main} for which the (weakened) bound $\sum_{j \in [n]}\omega_j C_j \Hquad \le \Hquad 2 \cdot \mathrm{OPT}$ holds.\\
Using the bound derived in the proof of Lemma \ref{lemma:cbfrelapprox}, we know that there exists algorithm $\mathrm{CBF}_R^\tau$ which can find a feasible coflow schedule for these deadlines such that for every $j \in [n]: \mathrm{CBF}_R^\tau(C_j) \le \tfrac{\tau+2}{\tau}C_j + 2\tau + 2$. Applying the algorithm to the deadlines yields:
\begin{align*}
\sum_{j \in [n]}\omega_j \cdot\mathrm{CBF}^\tau_R(C_j) &\Hquad\le\Hquad \sum_{j \in [n]}\omega_j \left( \frac{\tau+2}{\tau}C_j + 2\tau + 2 \right)\\
~&\Hquad=\Hquad \left(1 + \frac{2}{\tau}\right)\sum_{j \in [n]}\omega_j C_j \Hquad + \Hquad (2\tau + 2) \sum_{j \in [n]}\omega_j\\
~&\Hquad\le\Hquad \left(2 + \frac{4}{\tau} + \hat{\epsilon}(2\tau+2)\right) \mathrm{OPT}
\end{align*}
So for any $\epsilon$ by appropriate choice of $\tau$ large enough and respectively $\hat{\epsilon}$ small enough, the result follows.
\end{proof}

Note that while the requirements in Theorem \ref{thm:coflowapproxlarge} are rather technical, it implies several strong results for natural classes of coflow instances, such as instances where all coflows have large maximum degree.

\begin{corollary}
For any $\epsilon > 0$, there is $D \in \mathbb{N}$ such that there is a $(2+\epsilon)$-approximation algorithm for Coflow Scheduling without release dates for instances $\mathcal{I}$ fulfilling
\begin{equation*}
    \forall E_j \in \mathcal{I}:\quad \Delta(E_j) \Hquad\ge\Hquad D.
\end{equation*}
\end{corollary}
\begin{proof}
In any schedule, the finishing time of a coflow is lower bounded by its maximum degree. Therefore the optimum cost will be at least $D$ times the sum of the weights. So for $D$ large enough such that $\hat{\epsilon} \cdot D \ge 1$, Theorem \ref{thm:coflowapproxlarge} yields the result.
\end{proof}

%%%%%%%%%%%%%%%%%%%%%%%%%%%%%%%%%%%%%%%%%%%%%%%%%%%%%%%%%%%%%%%%%%%%%%%

%\bibliographystyle{alpha}
\printbibliography

%%%%%%%%%%%%%%%%%%%%%%%%%%%%%%%%%%%%%%%%%%%%%%%%%%%%%%%%%%%%%%%%%%%%%%%%%%%%%%%
\section{Appendix}\label{sec:app}
\subsection{Allocation LP Integrality and Complexity}\label{sec:app:lpintegrality}
\subsubsection*{Non-Integrality}
We provide a set of coflows and associated deadlines such that \ref{lp:main} contains a feasible fractional point but no integral point. The instance is defined on a vertex set $U \cup V$, where both $U$ and $V$ contain exactly $7$ vertices. It uses four coflows $E_1,E_2,E_3, E_4$ with associated deadlines $C_1 = 1, C_2 = 2, C_3 = 3, C_4 = 3$. The first three coflows act as a gadget construction which blocks certain vertices in $V$ from being used by edges in $E_4$. The edge sets of the gadget coflows are $E_1 = \{(6,2),(7,7)\}$,  $E_2 = \{(6,4),(7,5)\}$, $E_3 = \{(6,1),(7,3)\}$. Any valid integral schedule has to schedule $E_1$ in the first time slot, $E_2$ in the second time slot and $E_3$ in the third time slot, which essentially implies that coflow $E_4$ cannot use $V$ vertices $\{2\}, \{4,5\}, \{1,3\}$ in the first, second, and respectively third time slot. For easier argumentation we separate the edges of $E_4$ into two sets $E_4^1 = \{(1,1),(2,1),(3,3),(4,3),(1,4),(3,4),(2,5),(4,5)\}$ and $E_4^2 = \{(2,2),(3,2)\}$. There is a feasible half-integral allocation for this instance displayed in Figure \ref{fig:lpintegralitymatching}. 

\begin{figure}[H]
\centering
\includegraphics[width=15cm]{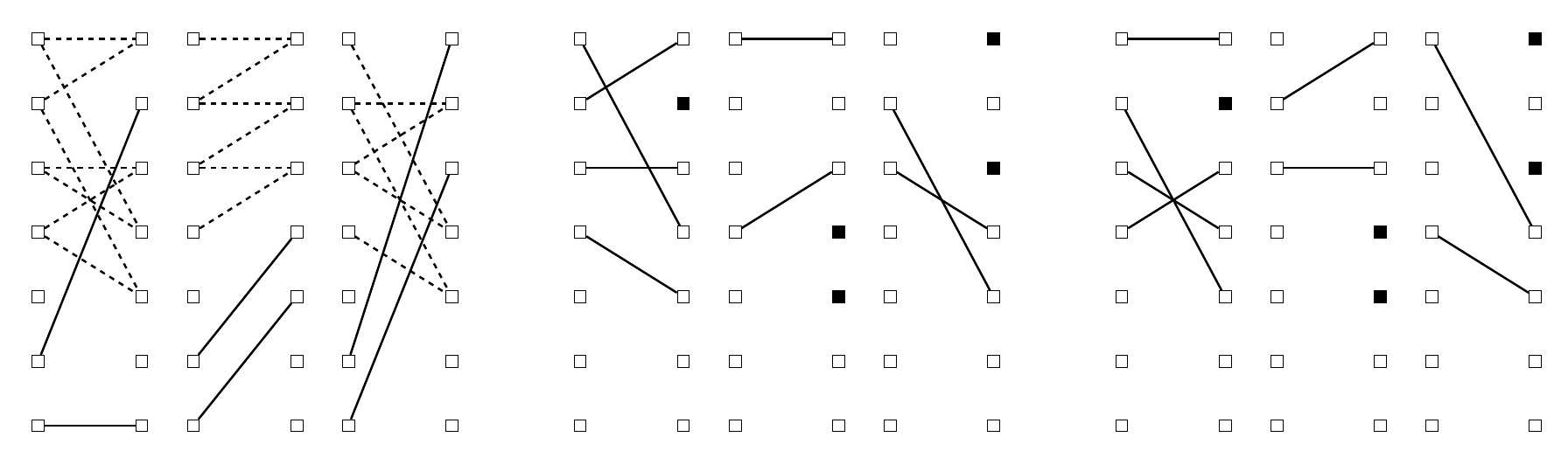}
\caption{The first picture shows a feasible fractional allocation, where dashed lines indicate values of $\tfrac{1}{2}$. The second and third picture show the two allocations for $E_1 \cup E_2 \cup E_3 \cup E_4^1$, with the relevant $V$ vertices which are blocked by gadget edges replaced by filled vertices for easier readability.}
\label{fig:lpintegralitymatching}
\end{figure}

By simple enumeration one can check that there are only two feasible integral allocations containing all edges in $E_1 \cup E_2 \cup E_3 \cup E_4^1$. The choice whether to schedule $(1,1)$ in the first slot and $(2,1)$ in the second slot or vice versa already fully determines a unique maximal matching in each case. However, for both of these allocations it is impossible to further include both edges in $E_4^2$, as vertex $2$ in $V$ is blocked in the first time slot and in both cases in either the second or the third time slot both vertices $2$ and $3$ are blocked in $U$. Therefore there is no integral matching containing all edges in $E_1 \cup E_2 \cup E_3 \cup E_4$.

\subsubsection*{$\mathbb{NP}$-Hardness}

For \ref{lp:main} it is possible to show that it is $\mathbb{NP}$-hard to decide whether an integral point exists. This is equivalent to showing that it is $\mathbb{NP}$-hard to determine whether a feasible integral schedule exists for some given coflow deadlines. The hardness holds even in a very restricted setting with just $3$ coflows and all deadlines in $\{1,2,3\}$.

\begin{lemma}\label{lemma:coflownphardness}
Given some bipartite multi-graph $G = (U \cup V, E)$, a disjoint partition of the edges $E = E_1 \cup \dotsc \cup E_n$ and deadlines $C_1,\dotsc,C_n \in \mathbb{N}$, it is $\mathbb{NP}$-complete to decide whether there exists a proper coloring $c: E \rightarrow \mathbb{N}$ such that for all $j \in [n], e \in E_j: c(e) \le d_j$.\\
This holds even for $n=3$ and $C_j \in \{1,2,3\}$.
\end{lemma}
\begin{proof}
We restate the proof found in \cite{schnaars_21}, it is closely inspired by work in \cite{bonuccelli_01}.\\
Containment in $\mathbb{NP}$ is clear, as checking whether the deadlines are fulfilled can be done in linear time. We reduce from a problem called Restricted Time Table Design (RTTD), which is known to be $\mathbb{NP}$-complete \cite{even_1975}.
\begin{lemma}[RTTD]\label{lemma:rttd}
Given sets $H:= \{h_1,h_2,h_3\}, P := \{P_1,...,P_n\}$, where each $P_j \subseteq H$, $C := [n]$, and a matrix $R \in \{0,1\}^{n \times n}$, it is $\mathbb{NP}$-hard to  decide whether there exists a function $f: P \times C \times H \rightarrow \{0,1\}$ such that:
\begin{itemize}
    \item[(1)] $f(P_j,c,h_k) = 1 \Hquad \Rightarrow \Hquad h_k \in P_j$
    \item[(2)] $\sum_{k=1}^3 f(P_j,c,h_k) = r_{j,c} \quad \forall j,c\in [n]$
    \item[(3)] $\sum_{j=1}^n f(P_j,c,h_k) \le 1 \quad \forall c \in [n], k \in [3]$
    \item[(4)] $\sum_{c=1}^n f(P_j,c,h_k) \le 1 \quad \forall j \in [n], k \in [3]$
\end{itemize}
\end{lemma}
RTTD can be seen as the problem of finding a set of matchings fulfilling certain constraints. We describe a reduction from an instance of RTTD to the graph coloring problem. We start with the graph $(P\cup C, E)$, where $E = \{(P_j,c) \ |\ \forall j,c \in [n]:\ r_{j,c} = 1\}$. Define a partition $E = E_1 \cup E_2 \cup E_3$, where for $i \in [3]: C_i := i$. Initially $E_1$ and $E_2$ are empty and $E_3$ contains all edges in $E$. Note that from a partition of $E_3$ into three matchings, which is equivalent to a valid $3$-edge coloring of $E_3$, one obtains a canonical function $f$ which fulfills conditions $(2)-(4)$. To ensure that $(1)$ is also fulfilled, for each $\hat{P} \in P$ we add certain edges. In case $h_1 \not\in \hat{P}$, for some new vertex $\hat{c} \in C$, we add the edge $(\hat{P},\hat{c})$ to $E_1$. In case $h_2 \not\in \hat{P}$, creating new vertices when necessary, we add $(\hat{v},\hat{c})$ to $E_1$ and $\hat{P},\hat{c}$ to $E_2$. In a similar manner, in the case $h_3 \not\in \hat{P}$, we add $(\hat{v},\hat{c})$ to $E_1$, $(\hat{v},\hat{c})$ to $E_2$ and $(\hat{P}, \hat{c})$ to $E_3$. In all cases, these additional edges can only be colored in a unique way, which causes the respective vertex $\hat{P}$ to be unavailable for other edges, thus enforcing constraint $(1)$.\\
Using this construction and considering the canonical mapping between edge colorings and decompositions into matchings, one obtains that a valid edge coloring respecting the deadlines in the graph is equivalent to a function $f$ in the sense of Lemma \ref{lemma:rttd}, implying the $\mathbb{NP}$-hardness of the problem.
\end{proof}

\subsection{Coflow Deadline Linear Program}\label{sec:app:coflowdeadline}
This section gives a short overview of the linear programming approach used to obtain deadlines for the coflows. All of the constructions and results are due to \cite{fukunaga22} and \cite{im19}.
\subsubsection*{Linear Program}
We use the linear program \ref{lp:deadlines} as given in \cite{fukunaga22}, adapted to our notation. For this purpose, let $E = \bigcup_{j \in [n]}E_j$ be the set of all edges and let $T = \max_{j \in [n]}r_j + 2\cdot \Delta(E)$ be an upper bound on the number of required time slots. Remember that we use the term flow to refer to an edge together with some possibly fractional multiplicity. The LP has one variable $x_{t,e}$ for each $t \in [T]$ and $e \in E$, which models the processing of flow $e$ during time step $t$. Additionally, for each coflow $j \in [n]$, there is a finishing time variable $c_j$.

\begin{gather*}\label{lp:deadlines}\tag{LP D}
\begin{aligned}
\min\qquad &&\sum_{j \in [n]} \omega_jc_j  &~&&~\\
s.t.\qquad &&\sum_{t \in [T]}t\cdot x_{t,e} &\quad\le c_j \qquad&\forall j \in [n], \forall e \in E_j&& (1)\\
~&&\sum_{e \in \delta_E(v)}x_{t,e} &\quad\le 1\qquad &\forall t \in [T],\forall v \in V && (2)\\
~&&\sum_{t \in [T]} x_{t,e} &\quad= 1\qquad &\forall j \in [n], \forall e \in E_j && (3)\\
~&&\phantom{\sum_{x \in [T]}}x_{t,e} &\quad =  0\qquad &\forall j \in [n], \forall e \in E_j, \forall t \in [r_j] && (4)\\
~&&x_{t,e} &\quad\ge 0\qquad &\forall t \in [T], \forall e \in E && (5)
\end{aligned}
\end{gather*}
\vspace{0.2cm}\\

Constraint (1) models that the completion time of each coflow is at least the time spent on scheduling each edge in said coflow. Constraint (2) models the matching constraints. Constraint (3) models that each edge of all coflows has to be fully scheduled.\\
Any valid solution for a coflow instance corresponds to a feasible point inside the polytope, so clearly we have $\mathrm{Cost}(\mathrm{LP}) \le \mathrm{OPT}$, where $\mathrm{OPT}$ is the optimal cost of the coflow instance. Note that any feasible solution to the LP forms a sequence of fractional matchings in the underlying graph. Instead of as a discrete fractional assignments in each time slot, we can also view them as continuous assignments during intervals. So if some edge $e$ during time $t$ has $x_{t,e} > 0$ flow assigned, we can view this as a continuous scheduling of $x_{t,e}$ flow amount during time $[t,t+1)$. In the same manner, given such a continuous assignment, by discretizing into unit length intervals, one can obtain a discrete fractional assignment. This view is both used in the coming rounding argument but also useful in general to analyze and understand the problem structure. 

\subsubsection*{Obtaining Coflow Deadlines}
The finishing time variables $c_j$ capture the notion of time spent processing each edge, but they do not correspond cleanly to deadlines for each coflow. If we for example set $C_j = \max\{t \in [T]: \exists e \in E_j: x_{e,t} > 0\}$, then the weighted sum over these deadlines might far exceed the cost of the LP, as it is possible that due to the fractionality of the matchings, some fraction of an edge is scheduled very late, even though the majority of the coflow is scheduled much earlier. To obtain deadlines for which there are strong approximation guarantees, we use a randomized rounding procedure as employed by \cite{fukunaga22,im19} and others.
For this purpose, for $j \in [n]$ and $\theta \in [0,1]$ let $C_j(\theta)$ be the smallest point in time at which all flows $e \in E_j$ have completed by at least a $\theta$ fraction in the continuous view of the fractional assignment. This means that we view each flow as being continuously assigned during the respective time slots and we determine the smallest point in time at which this continuous schedule reaches $\theta$ completion for all respective flows. The authors in the aforementioned works obtain deadlines by randomly selecting $\theta$ according to the probability distribution $f(x)=2x$ and setting $C'_j := \lceil C_j(\theta)/\theta\rceil$. As we do not require integrality in the next steps of our algorithm, we set $C_j := C_j(\theta)/\theta$, leaving out the rounding step.\\
It can be shown that the following holds for the rounded deadlines:

\begin{lemma}[\cite{im19}]\label{lemma:2approxold}
There is a polynomial time randomized algorithm determining deadlines $C'_1,\dotsc,C'_n$ for which \ref{lp:main} is feasible and for which for all $j \in [n]$:
\begin{equation*}
\mathbb{E}[C'_j] \ \le\  2c_j    
\end{equation*}
\end{lemma}

\noindent Slightly modifying their proof to account for not rounding the deadlines, we obtain:

\rslemmatwoapproxstrong*

\begin{proof}
In their Lemma 6, the authors show:
\begin{equation*}
    \sum_{j \in [n]}\omega_j \int_0^1 C_j(\theta)\  d \theta \Hquad = \Hquad \mathrm{Cost}(\mathrm{LP}) - \frac{1}{2}\sum_{j \in [n]}\omega_j 
\end{equation*}
By slight modification of the proof of their Lemma 9, we have:
\begin{align*}
    \mathbb{E}[\sum_{j \in [n]}\omega_j C_j] \Hquad&=\Hquad \mathbb{E}\Big[\sum_{j \in [n]} \frac{1}{\theta}C_j(\theta)\Big]\\
    ~ &=\Hquad 2 \sum_{j \in [n]} \omega_j \int_{\theta = 0}^1 C_j(\theta)\ d\theta\\
    ~ &=\Hquad 2\Big(\mathrm{Cost}(\mathrm{LP}) - \frac{1}{2}\sum_{j \in [n]}\omega_j\Big)\\
    ~ &\le\Hquad 2 \cdot \mathrm{OPT} - \sum_{j \in [n]}\omega_j \qedhere
\end{align*}

\end{proof}

Note that this procedure can be de-randomized to obtain a fully deterministic algorithm, for details see \cite{im19}.\\

The choice of $\theta$ and subsequent setting of $C_j := C_j(\theta)/\theta$ can be viewed as stretching the continuous schedule obtained for \ref{lp:deadlines}. This means that if some amount of flow $x_{e,t}$ is scheduled during time $[t,t+1)$, after the stretching the same amount of flow is scheduled during $[\frac{t}{\theta},\frac{t+1}{\theta})$. To not exceed the flow requirements, the schedule is cut off when the total amount of scheduled flow reaches the requirement. By definition of $C_j(\theta)$ and the matching constraints in $\ref{lp:deadlines}$, in the stretched schedule the matching constraints are still fulfilled and every coflow $E_j$ finishes by time $C_j$. By partitioning the continuous schedule into intervals between consecutive deadlines, the connection to \ref{lp:main} becomes clear. The maximum allocation constraints are fulfilled as the stretched schedule fulfills the matching constraints and by definition of $C_j$, for every edge enough flow is allocated before the deadline. This argument highlights that we can explicitly construct a feasible point for which \ref{lp:main} is feasible and the cost guarantees from Lemma \ref{lemma:2approxstrong} hold, by performing the stretching operation on the schedule obtained for \ref{lp:deadlines} and then discretizing the assignment with respect to each block.

%%%%%%%%%%%%%%%%%%%%%%%%%%%%%%%
\subsection{Coflow Scheduling with Release Dates}\label{sec:app:frameworkrelease}
In this section we show how the scheduling framework can be extended to work for the case with release dates. Theorem \ref{thm:coflowapproxrelease} gives an extension of the guarantee provided by Theorem \ref{thm:coflowapprox} to the case with release dates, though this comes at the cost of a worse approximation ratio.

\rsthmcoflowapproxrelease*

The overall proof structure is very similar to the case without release dates, just with tweaks at every step to account for the additional constraints. We provide the general outline here and omit some minor details which follow from modifications to the original arguments.

\subsubsection*{Coflow Deadlines}

Like in the case of no release dates, we want to obtain deadlines for the coflows which obey some structural constraints. \ref{lp:main} does not contain release dates, but with some minor modifications we obtain a suitable LP. For this purpose, for some $\kappa \in [2n]$, define a sequence $D_1 \le D_2 \le \dotsc \le D_\kappa$ containing exactly all deadlines and release dates. For some edge $e \in E$, let $r(e)$ be the index of the release date associated to $e$ in the chain of points in $D$ and respectively $d(e)$ the index of the deadline. Then we construct the following LP.

\begin{gather*}\tag{LP $R$}\label{lp:mainrel}
\begin{aligned}
\sum_{s \in [\kappa]}x_{s,e} &&=& \quad 1 &&\forall e \in E\\
\sum_{e: v \in e} x_{s,e} &&\le& \quad D_s - D_{s-1}\quad &&\forall s \in [\kappa],\forall v \in V\\
x_{s,e} &&=& \quad 0 &&\forall j \in [n], \forall e \in E_j, \forall s \not\in \{r(e)+1,\dotsc,d(e)\} \quad\\
x_{s,e} &&\ge&\quad 0
\end{aligned}
\end{gather*}

The structure of \ref{lp:mainrel} is very similar to \ref{lp:main}, just with added block separators for each release date and modification of the constraints to prevent edges from being scheduled in blocks before their respective release dates.\\

The same procedure by \cite{im19} used in Section \ref{sec:coflowdead} can be employed to obtain deadlines $C_1,\dotsc,C_n$ for which \ref{lp:mainrel} is feasible and for which the same cost bound from Lemma \ref{lemma:2approxstrong} holds.

\subsubsection*{Edge Allocation}

Given such a set of deadlines for which \ref{lp:mainrel} is feasible, we again describe two algorithms $\mathrm{Greedy}_R$ and $\mathrm{CBF}^\tau_R$ which provide feasible allocations for all edges. Their guarantees are slightly worse due to the added release date constraints.\\
Like previously, $\mathrm{Greedy}_R(C_j)$ and $\mathrm{CBF}^\tau_R(C_j)$ will be used to denote the finishing time of coflow $E_j$ in the schedule provided by the respective algorithm.

\begin{lemma}
For given deadlines $C_1,\dotsc,C_n$ for which \ref{lp:mainrel} is feasible there is an algorithm $\mathrm{Greedy}_R$ returning a valid coflow schedule such that the following holds for all  $j \in [n]$.
\begin{equation*}
\mathrm{Greedy}_R(C_j, r_j) \quad \le \quad r_j + 2C_j-1
\end{equation*}
\end{lemma}
\begin{proof}
The proof is essentially identical to the one of Lemma \ref{lemma:greedy}, just on a shifted interval. No edge of $E_j$ can be scheduled before $r_j$. For the following $2\lceil C_j \rceil$ time slots the same vertex allocation argument applies, which leads to an upper bound of $r_j + 2C_j - 1$.
\end{proof}

\noindent For the allocation procedure $\mathrm{CBF}^\tau_R$, the guarantee worsens by an additive $\tau+2$. 

\begin{lemma}\label{lemma:cbfrelapprox}
For given deadlines $C_1,\dotsc,C_n$ for which \ref{lp:mainrel} is feasible, weights $\omega_1,\dotsc,\omega_n$, and a parameter $\tau \in \mathbb{N}_{\ge 2}$, there is an algorithm $\mathrm{CBF}_R^\tau$ returning a valid coflow schedule such that the following holds.
\begin{equation*}
   \sum_{j \in [n]} \omega_j \cdot \mathrm{CBF}_R^\tau(C_j) \quad \le \quad \sum_{j \in [n]}\omega_j \left(\frac{\tau+2}{\tau}C_j + \frac{3}{2}\tau + 4.5 - \frac{2}{\tau}\right)
\end{equation*}
\end{lemma}
\begin{proof}
The overall algorithm is almost identical to the one for the case of no release dates from Section \ref{sec:beckfiala}. The main change is a different initial rounding.\\

Note that in order to obtain good guarantees, we need to ensure that the number of resulting blocks after rounding is not too large and that we have control over the blocks' sizes. We therefore have to round both release dates and deadlines. Simply rounding both to the next multiple of $\tau$ would not suffice, as it could lead to infeasible LPs. For example, if for some $k\in \mathbb{N}$, $r_j = k \cdot \tau + 1$ and $C_j = (k+1)\cdot \tau-\epsilon$, then rounding them in that way would lead to $r_j = C_j = (k+1)\cdot \tau$. We instead round the release dates up to the next multiple of $\tau$ and the deadlines to the second next multiple of $\tau$, meaning that we round them to the next multiple and then add an additional $\tau$.\\

In the case without release dates, it is not hard to show that rounding up the deadlines cannot make the resulting LP infeasible, as the feasible region only increases. In the present case, a bit more care is needed, as the rounding of the release dates could lead to parts of the feasible region becoming infeasible. A feasible point for the original LP can be transformed to one in this LP by interpreting the assignment as a time-continuous one and essentially shifting the allocation by $\tau$. This means that if an edge was scheduled at some point in time $t$, we now treat it as if it was scheduled at time $t+\tau$. More details about these transformations and interpretations can be found in \cite{fukunaga22}.\\

Given the feasibility of the LP for the rounded release dates and deadlines, the same iterated rounding approach from Section \ref{sec:beckfiala} can be used to obtain an integral feasible schedule which violates the respective degree bound constraints by at most $2$. As the procedure never changes variables as soon as they are integral and as the blocks's sizes increasing only increases assigned time slots, the feasibility for the release date constraints is preserved.

For some given deadline $C_j$, let $k \in \mathbb{N}$ and $a \in [0,\tau)$ such that $C_j = k\cdot \tau + a$. Then $C_j$ gets rounded to $(k+2)\cdot \tau$. There are at most $k+2$ blocks up to and including the block formed by $C_j$, whose sizes all increase by at most two. This yields the following bound.

\begin{equation*}
    \mathrm{CBF}^\tau_R(C_j) \Hquad\le\Hquad (k+2)\tau + (k+2)2 \Hquad=\Hquad (k+2)(\tau+2) \Hquad\le\Hquad \frac{\tau+2}{\tau}C_j + 2\tau + 4
\end{equation*}
Like in the case of no release dates, this bound is slightly weaker than as stated in the lemma. Using the same averaging strategy as employed before reduces the additive constant by $\tau/2 - \tfrac{1}{2} + \frac{2}{\tau}$, leading to the result. 
\end{proof}

\subsubsection*{Framework and Approximation Bound}
The algorithm for Coflow Scheduling with release dates again works by obtaining deadlines and then running several edge allocation algorithms on these and returning the cheapest solution among them. To bound the cost, a framework very similar to the one described in Lemma \ref{lemma:appbound} is used, though an additional bound on the distance to the optimum cost is needed due to the presence of the $r_j$ summand in the guarantee provided by $\mathrm{Greedy}_R$. As in any optimal solution the finishing time $\mathrm{OPT}_j$ of coflow $E_j$ has to be after $r_j$, we have $r_j \le \mathrm{OPT}_j - 1$.\\

In the following lemma, like in Lemma \ref{lemma:appbound}, let $f_1,\dotsc,f_k$ be some functions capturing the edge allocation guarantees provided by some collection of algorithms $\mathrm{ALG}_1,\dotsc,\mathrm{ALG}_k$. In this case the functions additionally depend on a parameter $r_j \in \mathbb{R}_{\ge 0}$, which like in the case for $\mathrm{Greedy}$ captures the dependency on release dates.

\begin{lemma}\label{lemma:appboundrelease}
Let $\lambda_1,\dotsc,\lambda_k \ge 0$ with $\sum_{i \in [k]}\lambda_i = 1$ and $a,b \in \mathbb{R}$. If for all possible pairs $x \ge 1, r_x \in [0,x-1]$
\begin{equation*}
    \sum_{i \in [k]}\lambda_i f_i(x, r_x) \quad \le \quad a(x+1) + b(r_x+1),
\end{equation*}
then for all coflow instances $\mathcal{I}$:
\begin{equation*}
    C_{\mathrm{ALG}}(\mathcal{I}) \Hquad=\Hquad \min\{C_{\mathrm{ALG}_1}(\mathcal{I}),\dotsc,C_{\mathrm{ALG}_k}(\mathcal{I})\} \Hquad\le\Hquad (2a + b) \cdot \mathrm{OPT}(\mathcal{I})
\end{equation*}
\end{lemma}
\begin{proof}
Let $g(x) = \sum_{i \in [k]}\lambda_i f_i(x,r_x)$. Using the exact same argument as in the proof of Lemma \ref{lemma:appbound} and inserting the upper bound on $g$, we arrive at
\begin{equation*}
    C_{ALG} \quad \le \quad \sum_{j \in [n]}\omega_j g(C_j) \quad \le \quad a \sum_{j \in [n]}\omega_j C_j \ +\  a\sum_{j \in [n]}\omega_j \ +\  b \sum_{j \in [n]}\omega_j(r_j+1)
\end{equation*}
Inserting $\sum_{j \in [n]}\omega_j C_j \quad \le \quad 2 \cdot \mathrm{OPT} - \sum_{j \in [n]}\omega_j$ for the first sum and $r_j \le \mathrm{OPT}_j-1$ into the second sum we obtain
\begin{equation*}
C_{\mathrm{ALG}} \quad \le \quad (2a+b) \cdot \mathrm{OPT} \qedhere
\end{equation*}
\end{proof}
\noindent We can now apply this modified framework to show Theorem \ref{thm:coflowapproxrelease}.

\rsthmcoflowapproxrelease*
\begin{proof}
We use algorithms $\mathrm{Greedy}_R$ and $\mathrm{CBF}_R^4$. For their edge allocation guarantees we have
\begin{equation*}
    f_G(C_j, r_j) \Hquad\le\Hquad r_j + 2C_j - 1\qquad \text{and} \qquad f_{\mathrm{CBF}}(C_j, r_j) \Hquad\le\Hquad \frac{3}{2}C_j + 10.
\end{equation*}
For $\lambda_1 = 0.68$ and $\lambda_2 = 0.32$ we obtain
\begin{equation*}
    \lambda_1 f_G(x, r_x) + \lambda_2 f_{\mathrm{CBF}}(x, r_x) \quad\le\quad 1.84 (x+1) \Hquad+\Hquad 0.68(r_j+1),
\end{equation*}
which by application of Lemma \ref{lemma:appboundrelease} gives a $4.36$-approximation.
\end{proof}

%%%%%%%%%%%%%%%%%%%%%%%%%%%%%%%
\subsection{High Edge Multiplicities}\label{sec:app:coflowpseudopoly}

In the main body of this work we have assumed that in each coflow, each edge is given explicitly, meaning that multiple copies of said edge are included if the respective flow demand is greater than $1$. In this chapter we show how to extend the algorithms to the setting where instead each edge $e\in E$ has some associated flow requirement $p_e \in \mathbb{N}_+$. To ensure polynomial runtime, we need dependency on $\mathcal{O}(\log(p_e))$ instead of $\mathcal{O}(p_e)$. We individually highlight the changes required in each step. Extending the algorithms to this setting increases the cost by a multiplicative factor of $(1+\epsilon)$, for an arbitrarily small $\epsilon > 0$.

\subsubsection*{Coflow Deadlines}

Obtaining coflow deadlines in this modified setting requires some additional care, as the procedure uses a time-indexed LP, which could potentially require a super-polynomial amount of time slots. We adapt the procedure described in \cite{im19} to our case. The basic idea is to change the time-indices to interval indices, for a polynomially sized set of geometrically increasing intervals. We focus on the case without release dates, the case with them follows analogously.\\

Let $T$ be some upper bound on the number of required time slots, for example the sum over all multiplicities. For some $\epsilon > 0$, define the set of points $\{\lfloor (1+\epsilon)^i\rfloor\}_{i \in [\lceil \log_{1+\epsilon}T\rceil]}$. For some $K \in \mathbb{N}$, let these timepoints be $1 = t_1 \le t_2 \le \dotsc \le t_K$. For convenience, set $t_0 = 0$. The modified LP is as follows.

\begin{gather*}\label{lp:deadlinespoly}\tag{LP D'}
\begin{aligned}
\min\qquad &&\sum_{j \in [n]} \omega_jc_j  &~&&~\\
s.t.\qquad &&\sum_{i \in [K]}(t_{i+1}-1)\cdot x_{i,e} &\quad\le c_j \qquad&\forall j \in [n], \forall e \in E_j&& (1)\\
~&&\sum_{e \in \delta_E(v)}x_{t,e} &\quad\le t_i - t_{i-1}\qquad &\forall i \in [K],\forall v \in V && (2)\\
~&&\sum_{i \in [K]} x_{i,e} &\quad= p_e\qquad &\forall j \in [n], \forall e \in E_j && (3)\\
~&&x_{i,e} &\quad\ge 0\qquad &\forall i \in [K], \forall e \in E && (4)
\end{aligned}
\end{gather*}

In this LP, the timepoints are replaced by time intervals of size $t_{i+1} - t_i$. The right hand sides of constraints $(2)$ and $(3)$ are respectively adjusted to account for this. If an interval contains some amount of flow $x_{i,e}$, it is assumed that this is scheduled instantaneously at the latest possible time point $t_{i+1}-1$ in the interval, which gives the term in $(1)$. This leads to an overestimation of the cost by a factor of at most $(1+\epsilon)$. The rounding procedure to determine deadlines $C_1,\dotsc,C_n$ and the associated results for LP feasiblity remain unchanged.

\subsubsection*{Greedy Flow Allocation}

For the greedy algorithm, we cannot simply do one iteration for each edge and multiplicity, as this could require super-polynomially many steps. The required adaption to the procedure was described in a similar form in \cite{zhen_2015}.\\

We create $n$ sets of edges iteratively, one associated to each coflow. Assume that $E_1,\dotsc,E_n$ are the sets of edges ordered without loss of generality ascendingly by the deadlines returned in the first step. The first set initially only contains the edges in $E_1$. When creating the $j$-th set, we iterate over all edges $e \in E_j$. If there exists some set with index $i < j$ such that $e$ could be added to the set without increasing its maximum vertex degree, we add a copy $\hat{e}$ of $e$ to this set with as much flow demand $p_{\hat{e}}$ as possible without increasing the maximum degree and remove this flow demand from $p_e$. After doing this procedure for all coflows and edges, we have $n$ sets of edges, some of which might be empty. Each of these blocks can now be scheduled consecutively in polynomial time using Theorem \ref{thm:koenig} (Kőnig's Theorem). It is easy to check that for each edge the greedy property is fulfilled, meaning it cannot be scheduled earlier without changing other allocation. Therefore the bounds guaranteed by Lemma \ref{lemma:greedy} hold. In the case of release dates, the sets have to be subdivided to prevent crossing any release date and during the flow shifting only sets have to be considered for which the lowest possible time slot is larger than the respective release date.

\subsubsection*{Coflow Beck-Fiala}

To obtain a polynomial algorithm for the $\mathrm{CBF}^\tau$ allocation procedure we only need to do a minor modification to \ref{lp:cbf}. Initially, in constraint $(I)$, we change the right hand side to $p_e$ instead of $1$. After having obtained a solution to this modified LP, for every edge and every block, fix the integral part of the respective flow assigned to this block. Replace each edge $e$ by $\hat{p}_e$ copies and change the right hand side to $1$ again, where $\hat{p}_e$ is the sum over the remaining fractional assignments. As there are at most $2n$ blocks, this amount is upper bounded by $2n$, hence strongly polynomial in the input size. Solve the remaining instance with the unmodified iterated rounding process.

\subsection{LP Integrality Gap}\label{sec:app:lpintgap}
In his work, using a non-constructive result from hypergraph matching theory, Fukunaga \cite{fukunaga22} shows that the integrality gap of \ref{lp:deadlines} is at most $4$, even in the case of release dates. By using the framework from Lemma \ref{lemma:appbound} on a slightly refined version of the non-constructive hypergraph result together with one of the edge allocation algorithms from this work, we establish a stronger bound.

\begin{lemma}\label{lemma:lpintgap}
The integrality gap of \ref{lp:deadlines} is at most $\frac{109}{28}(<3.893)$.
\end{lemma}

\noindent To show this, we use the following bound which is obtained by refining a result from \cite{fukunaga22}.

\begin{lemma}\label{lemma:hypergraphexist}
Given deadlines $C_1,\dotsc,C_n$ and release dates $r_1,\dotsc,r_n$ for which \ref{lp:mainrel} is feasible, there exists a valid coflow schedule such that for the finishing times $C^*_j$ the following holds.
\begin{equation*}
    C^*_j \quad \le \quad 2C_j + 1
\end{equation*}
\end{lemma}
\begin{proof}

We briefly review the technique used by \cite{fukunaga22} before explaining how to achieve the refined result. For an elaborate proof see \cite{fukunaga22}. Their proof uses the following result from hypergraph matching theory by Aharoni and Haxell.

\begin{lemma}[\cite{Aharoni2000}]\label{lemma:hypergraphmatching}
    If an $r$-uniform bipartite hypergraph $\mathcal{H}$ with the node set bipartition $(A,B)$ satisfies $v(H_X) > (r-1)(|X|-1)$ for any $X \subseteq A$, then $H$ has a perfect matching.
\end{lemma}

A hypergraph is called $r$-uniform if every edge contains exactly $r$ vertices and is called bipartite if there is a bipartition $(A,B)$ of the vertices such that for all edges $e: |e \cap A| = 1$. Here $H_X$ is the induced edge set in $B$ for some given set of vertices $X \subseteq A$, meaning that $H_X := \{e \setminus A\ |\ e \in E , e \cap X \neq \emptyset\}$ and $v(E)$ is the maximum size of a matching in a given set of edges $E$.\\

Given the deadlines $C_1,...,C_n$ and release dates $r_1,...,r_n$, we construct a $3$-uniform bipartite hypergraph for which we can show existence of a perfect matching using Lemma \ref{lemma:hypergraphmatching}. This perfect matching can be directly translated to a feasible solution of the Coflow Scheduling instance. For this purpose, let $T$ be an upper bound on the number of required time slots and for every $t \in [T]$ let $V^1_t$ and $V^2_t$ be copies of the node set of the coflow instance. For $v \in V, t\in [T],k \in [2]$, let $v_t^k$ be the respective corresponding vertex in $V_t^k$. For every $e \in E$, let $a_e$ be a vertex. The bipartition is given by $A := \{a_e\ |\ e \in E\}$ and $B:= \cup_{t \in [T], k \in [2]}V^k_t$. For every $j \in [n], e=(v,u) \in E_j, t \in \{r_j+1,...,\lceil C_j \rceil\}$ we define the two hyperedges $e_t^1 := \{a_e, v_t^1, u_t^1\}$ and $e_t^2 := \{a_e, v_t^2, u_t^2\}$.\\

Given some $X \subseteq A$, the set $H_X$ corresponds to disjoint pairs of copies of the underlying graph for all time slots in which the respective edges belonging to vertices in $X$ are scheduled. As this is a bipartite graph and as the deadlines and release dates were feasible for \ref{lp:mainrel}, there is a matching of size $2|X|$ in this graph. This follows by interpreting the fractional matching in the polytope as a fractional matching in both of the vertex sets. Hence the requirements of Lemma \ref{lemma:hypergraphmatching} are fulfilled, so there is a perfect matching in the hypergraph. Such a matching directly corresponds to an assignment of edges to time slots, where the two sets $V_t^1$ and $V_t^2$ correspond to each time slot essentially being doubled, so being replaced by two consecutive time slots. In the resulting schedule, for every $E_j$, each of its edges is scheduled by time $2\lceil C_j \rceil \le 2C_j +2$. \\

With a small modification, this result can be improved by $1$. For the largest value of $t$ in $\{r_j+1,...,\lceil C_j \rceil\}$, if the fractional part of $C_j$ is less than or equal to $\frac{1}{2}$, we do not include the hyperedge $e_t^2$. In this case, by considering the canonical discretization of the edge allocation from the point in \ref{lp:mainrel}, from coflows with deadlines less than or equal to $C_j$ there is at most a total assignment of $\frac{1}{2}$ adjacent to the vertices in the time slot $\lceil C_j \rceil$. Hence this assignment can be doubled and assigned to the respective first hyperedge, with remaining allocation from other coflows possibly being pushed to the respective second hyperedge. So in this case the finishing time bound improves by $1$, while for fractional part greater than $\frac{1}{2}$, we have $2\lceil C_j \rceil \le 2C_j +1$.
\end{proof}

\noindent Combining this with the guarantees obtained from Lemma \ref{lemma:cbfrelapprox}, Lemma \ref{lemma:lpintgap} can be shown.
\begin{proof}[Proof of Lemma \ref{lemma:lpintgap}]
Note that the framework from Lemma \ref{lemma:appbound} only requires guarantees for the edge allocations of the given algorithms and is thus still applicable even if some of them might be non-constructive. The conclusion then changes to existence of a solution fulfilling the given cost bound, rather than a constructive algorithm, but this is sufficient to establish a bound on the integrality gap. We use a combination of the guarantees provided for the non-constructive result from Lemma \ref{lemma:hypergraphexist} and for $\mathrm{CBF}_R^5$ from Lemma \ref{lemma:cbfrelapprox}. Let $f_H$ and $f_{CBF}$ be functions capturing the respective guarantees in the sense required by Lemma \ref{lemma:appbound}. We have 
\begin{equation*}
    f_H(x) \Hquad = \Hquad 2x + 1 \quad \text{and} \quad f_{\mathrm{CBF}}(x) \Hquad = \Hquad \frac{7}{5}x + 11.6.
\end{equation*}
Let $\lambda_1 = \frac{51}{56}$ and $\lambda_2 = \frac{5}{56}$. Then we obtain
\begin{equation*}
\lambda_1 f_H(x) + \lambda_2f_{\mathrm{CBF}}(x) \Hquad=\Hquad (2 \cdot \tfrac{51}{56} + \tfrac{7}{5}\cdot \tfrac{5}{56})x + (\tfrac{51}{56} + 11.6 \cdot \tfrac{5}{56}) \Hquad=\Hquad \tfrac{109}{56}(x+1).
\end{equation*}
So the requirements of Lemma \ref{lemma:appbound} are fulfilled with $\alpha = \frac{109}{56}$ and we thus obtain that there exists a valid coflow schedule with cost at most $2 \cdot \frac{109}{56} = \frac{109}{28} (< 3.893)$, which implies that the integrality gap of LP \ref{lp:deadlines} is at most this value.
\end{proof}

\subsection{Approximation Improvements}\label{sec:app:approximprov}

Using the framework from Lemma \ref{lemma:appbound} together with a new edge allocation function, we can achieve a slight improvement upon the $3.415$-approximation from Theorem \ref{thm:coflowapprox}.\\

\begin{lemma}\label{label:app:koenigcbf}
Given deadlines $C_1,\dotsc,C_n$ for which \ref{lp:cbf} is feasible, weights $\omega_1,...,\omega_n$ and a parameter $\tau \in \mathbb{N}_{\ge 2}$ and a parameter $b \in \mathbb{N}_{\ge 1}$, there is an algorithm $\mathrm{CKBF}^{(\tau,b)}$ returning a valid coflow schedule such that the following holds.
\begin{equation*}
    \sum_{j \in [n]}\omega_j \cdot\mathrm{CKBF}^{(\tau,b)}(C_j) \Hquad\le\Hquad b \cdot \sum_{j: C_j < b+1} \omega_j \Hquad + \Hquad \sum_{j: C_j \ge b+1}\omega_j\left(\frac{\tau+2}{\tau}C_j + \frac{\tau}{2} + 2.5 + b - \frac{2}{\tau}\right)
\end{equation*}
\end{lemma}
\begin{proof}

The algorithm is a combination of Kőnig's Theorem with the $\mathrm{CBF}^\tau$ allocation procedure.\\
Let $E(b) := \cup_{j:C_j < b+1}E_j$ the set of edges contained in coflows with deadline strictly smaller than $b+1$. We claim that $\Delta(E(b)) \le b$. Assume otherwise, then there has to be a vertex with degree greater than or equal to $b+1$. But this would violate the feasibility of \ref{lp:main} for the deadlines, as there is no way to allocate $b+1$ or more edges adjacent to the same vertex within strictly less than $b+1$ time.\\
For the algorithm, take all coflows with deadline strictly smaller than $b+1$ and schedule them within the first $b$ time-slots. Schedule all remaining coflows using $\mathrm{CBF}^\tau$, shifting their allocation by $b$ time-slots. The claimed allocation guarantee is immediate from the definitions of Kőnig's Theorem and $\mathrm{CBF}^\tau$.
\end{proof}

Combining this algorithm with the edge allocation algorithms previously used, the following stronger bound can be derived. The difference between the two bounds is slightly more than $\frac{1}{100}$.

\begin{lemma}
There is a polynomial time algorithm achieving a $\frac{497}{146}(< 3.4042)$-approximation for Coflow Scheduling without release dates.
\end{lemma}
\begin{proof}
We use a combination of $\mathrm{Greedy}$, $\mathrm{CBF}^\tau_\lambda$, and $\mathrm{CKBF}^{(\tau,b)}$. Let their guarantees be given by $f_G, f_{\mathrm{CBF}^\tau_\lambda}$, and $f_{\mathrm{CKBF}^{(\tau,b)}}$. We know that:
\begin{align*}
    f_G(x) \quad&=\quad 2x-1\\
    f_{\mathrm{CBF}^\tau_\lambda}(x) \quad&=\quad \begin{cases}\lambda + 2 & x \le \lambda\\(\tau+2) \cdot \lceil \frac{x}{\tau} \rceil &  \lambda = 0\\ (\tau+2) \cdot \lceil \frac{x - \lambda}{\tau} \rceil + (\lambda + 2) & \lambda \neq 0 \wedge x > \lambda  \end{cases}\\
    f_{\mathrm{CKBF}^{(\tau,b)}}(x) \quad&=\quad \begin{cases}b & x < b+1\\b + f_{\mathrm{CBF}^\tau}(x) & x \ge b+1\end{cases}
\end{align*}
Define $\lambda_1 := \frac{749}{1460}$, $\lambda_2 := \frac{126}{1460}$, and $\lambda_3 := \frac{117}{1460}$. Then we have:
\begin{equation*}
    \lambda_1 \cdot \mathrm{Greedy}(x) \Hquad+\Hquad \lambda_2 \cdot \mathrm{CKBF}^{(6,1)}(x) \Hquad+\Hquad \lambda_3 \cdot \sum_{\lambda \in \{0,3,4,6,7\}}\mathrm{CBF}^5_\lambda(x) \quad \le \quad \tfrac{2485}{1460} (x+1)
\end{equation*}
Thus from Lemma \ref{lemma:appbound} we obtain a $2 \cdot \frac{2485}{1460} = \frac{497}{146}$-approximation for Coflow Scheduling without release dates.
\end{proof}

Note that the same approach does not give an improvement for neither Coflow Scheduling with release dates nor the integrality gap from Section \ref{sec:app:lpintgap}. We surmise that similar improvements for these cases and also further tiny improvements for the case without release dates might be possible using more involved combinations of (possibly new) edge allocation functions.
\end{document}